%% file: main.tex
\documentclass{article}
\usepackage{booktabs}
\usepackage{pgfplots}
\usepackage{subfigure}
\usepackage{calc}

\usepackage{xspace}
\usepackage{enumerate}

\usepackage[cmex10]{amsmath}
\usepackage{amsfonts,amssymb,amsthm,amstext,latexsym, dsfont}

\usepackage{algorithm, algorithmicx, algpseudocode}

\newcommand{\vtxt}[1]{\mathit{#1}}

\newcommand{\ema}[1]{\ensuremath{#1}}
\newcommand{\fff}{\ema{f}}
\newcommand{\impset}[1]{\mathds{#1}}
\newcommand{\N}[1]{\impset{N}^{#1}}

\newcommand{\R}[1]{\impset{R}^{#1}}

\newcommand{\mO}[1]{\mathcal{O}\left(#1\right)}

\newcommand{\ep}[1]{\mathbb{E}\left(#1\right)}
\newcommand{\et}[1]{E_t\left(#1\right)}

\newcommand{\norm}[1]{\left\| #1\right\|}
\newcommand{\fl}{fl}

\newcommand{\vtr}[1]{\mathbf{#1}}
\newcommand{\trans}[1]{#1^\intercal}

\newcommand{\np}{n^\prime}

\newcommand{\eps}{\varepsilon}

\DeclareMathOperator*{\rk}{rk}
\DeclareMathOperator*{\nnz}{nnz}
\DeclareMathOperator*{\argmin}{argmin}

\newcommand{\onlinechen}{\textsc{Online-Detec\-tion}\xspace}
\newcommand{\abftdetect}{\textsc{ABFT-Detec\-tion}\xspace}
\newcommand{\abftcorrect}{\textsc{ABFT-Correc\-tion}\xspace}

\algrenewcommand\algorithmicindent{2em}%

\algtext*{EndFor}
\algtext*{EndIf}
\algtext*{EndFunction}

\algrenewcommand\Return{\State \algorithmicreturn{} }%

\newtheorem{theorem}{Theorem}

\title{A Backward/Forward Recovery Approach
for the Preconditioned Conjugate Gradient Method}

\author{Massimiliano Fasi, 
Julien Langou,
Yves Robert, and
Bora U\c{c}ar}

\date{}

\begin{document}

\maketitle

\begin{abstract}
    Several recent papers have introduced a periodic verification mechanism to
    detect silent errors in iterative solvers. Chen [PPoPP'13,  pp. 167--176] has
    shown how to combine such a verification mechanism (a stability test checking
    the orthogonality of two vectors and recomputing the residual)
    with checkpointing: the idea is to verify every $d$ iterations, and to
    checkpoint every $c \times d$ iterations. When a silent error is detected by
    the verification mechanism, one can rollback to and re-execute from the last
    checkpoint.
    In this paper, we also propose to combine checkpointing and verification, but
    we use algorithm-based fault tolerance (ABFT) rather than stability tests. ABFT can be used for error detection,
    but also for error detection and correction, allowing a forward recovery (and
    no rollback nor re-execution) when a single error is detected. We introduce an
    abstract performance model to compute the performance of all schemes, and we
    instantiate it using the preconditioned conjugate gradient algorithm.
    Finally, we validate our new approach through a set of simulations.
\end{abstract}

\section{Introduction}

Silent errors (or silent data corruptions) have become a significant concern in
HPC environments~\cite{moody2010}.
There are many sources of silent errors, from bit flips in cache caused by
cosmic radiations, to wrong results produced by the arithmetic logic unit.
The latter source becomes relevant when the computation is performed
in the low voltage mode%, as it is
%often the case for large-scale computations.
%While the low voltage mode reduces the energy consumption,
to reduce the energy consumption in  large-scale computations.
But the low levels of voltage
dramatically reduces the dependability of the system.

The key problem with silent errors is the \emph{detection latency}:
when a silent error strikes, the corrupted data is not identified immediately,
but instead only when some numerical anomaly is detected in the application behavior.
It is clear that this detection can occur with an arbitrary delay.
As a consequence, the de facto standard method for resilience, namely checkpointing
and recovery, cannot be used directly.
Indeed, the method of checkpointing and recovery applies to fail-stop errors (e.g., hardware
crashes): such errors are detected immediately, and one can safely recover from
the last saved snapshot of the application state.
On the contrary, because of the detection latency induced by silent errors, it
is often impossible to know when the error struck, and hence to determine which
checkpoint (if any) is valid to safely restore the application
state.
Even if an unlimited number of checkpoints could be kept in memory, there would
remain the problem of identifying a valid one.

In the absence of a resilience method, the only known remedy to silent errors is
to re-execute the application from scratch as soon as a silent error is
detected.
On large-scale systems, the silent error rate grows linearly with the number of
components, and several silent errors are expected to strike during the
execution of a typical large-scale HPC application~\cite{cappello2009,FranckCappello11012009,cappello14-resilience,Schroeder:2007}.
The cost of re-executing the application one or more times becomes prohibitive,
and other approaches need to be considered.

Several recent papers have proposed to introduce a verification
mechanism to be applied periodically in order to detect silent errors.
These papers mostly target iterative methods to solve sparse linear systems,
which are natural candidates to periodic detection.
If we apply the verification mechanism every, say, $d$ iterations, then we have
the opportunity to detect the error earlier, namely at most $d-1$ iterations
after the actual faulty iteration, thereby stopping the progress of a flawed
execution much earlier than without detection.
However, the cost of the verification may be \mbox{non-negligible} in front of the cost
of one iteration of the application, hence the need to trade off for an adequate value
of $d$.
Verification can consist in testing the orthogonality of two vectors (cheap) or
recomputing the residual (cost of a sparse matrix-vector product, more
expensive).
We survey several verification mechanisms in Section~\ref{sec.related}.
Note that in all these approaches
%apply only to a given type of silent
%errors: in other words,
a \emph{selective reliability} model is enforced, where
the parts of the application that are not protected are assumed to execute in a
reliable mode.

While verification mechanisms speed up the detection of silent errors, they
cannot provide correction, and thus they cannot avoid the re-execution of the
application from scratch.
A solution is to combine checkpointing with verification.
If we apply the verification mechanism every $d$ iterations, we can checkpoint
every $c \times d$ iterations, thereby limiting the amount of re-execution
considerably.
A checkpoint is always valid because it is being preceded by a verification.
If an error occurs, it will be detected by one of the $c$ verifications performed
before the next checkpoint.
This is exactly the approach proposed by Chen~\cite{Chen2013} for a variety of
methods based on Krylov subspaces, including the widely used conjugate Gradient
(CG) algorithm.
Chen~\cite{Chen2013} gives an equation for the overhead incurred by
checkpointing and verification, and determines the best values of $c$ and $d$ by
finding a numerical solution of the equation.
In fact, computing the optimal verification and checkpoint intervals is a hard
problem.
In the case of pure periodic checkpointing,  closed-form approximations of the
optimal period have been given by Young~\cite{young74} and Daly~\cite{daly04}.
However, when combining checkpointing and verification, the complexity of the
problem grows.
To the best of our knowledge, there is no known closed-form formula, although a
dynamic programming algorithm to compute the optimal repartition of checkpoints
and verifications is available~\cite{pmbs2014}.

For linear algebra kernels, another widely used technique for silent error
detection is algorithm-based fault tolerance (ABFT).
The pioneering paper of Huang and Abraham~\cite{abft-matrix-operations}
describes an algorithm capable of detecting and correcting a single silent error
striking a dense matrix-matrix multiplication by means of row and column
checksums.
ABFT protection has been successfully applied to dense LU~\cite{pengduppopp12}, LU with 
partial pivoting~\cite{Yao01112015}, Cholesky~\cite{Hakkarinen2015} and QR~\cite{Du2013}
factorizations, and more recently to sparse kernels like SpMxV (matrix-vector
product) and triangular solve~\cite{Shantharam2012}.
The overhead induced by ABFT is usually small, which makes it a good candidate
for error detection at each iteration of the CG algorithm.

The beauty of ABFT is that it can \emph{correct} errors in addition to detecting
them.
This comes at the price of an increased overhead, because several checksums are
needed to detect and correct, while a single checksum is enough when just
detection is required.
Still, being able to correct a silent error on the fly allows for \emph{forward
recovery}.
No rollback, recovery nor re-execution are needed when a single silent error is
detected at some iteration, because ABFT can correct it, and the execution can
be safely resumed from that very same iteration.
Only when two or more silent errors strike within an iteration we do need
to rollback to the last checkpoint.
In many practical situations, it is reasonable to expect no more than one error
per iteration, which means that most  roll-back operations can be avoided.
In turn, this leads to less frequent checkpoints, and hence less overhead.

The major contributions of this paper are an ABFT framework to detect multiple
errors striking the computation and a performance model that allows to compare
methods that combine verification and checkpointing.
The verification mechanism is capable of error detection, or of both error
detection and correction.
The model tries to determine the optimal intervals for verification and
checkpointing, given the cost of an iteration, the overhead associated to
verification, checkpoint and recovery, and the rate of silent errors.
Our abstract model provides the optimal answer to this question, as a function
of the cost of all application and resilience parameters.

We instantiate the model using a CG kernel, preconditioned with a sparse
approximate inverse~\cite{Benzi1998}, and compare the performance of two
\mbox{ABFT-based} verification mechanisms.
We call the first scheme, capable of error detection only, \abftdetect and
the second scheme, which enhances the first one by providing single error correction
as well, \abftcorrect.
Through numerical simulations, we compare the performance of both schemes with
\onlinechen, the approach of Chen~\cite{Chen2013} (which we extend to recover
from memory errors by checkpointing the sparse matrix in addition to the current
iteration vectors).
These simulations show that \abftcorrect outperforms both \onlinechen and
\abftdetect for a wide range of fault rates, thereby demonstrating that
combining checkpointing with ABFT correcting techniques is more efficient than
pure checkpointing for most practical situations.

Our discussion focuses on the sequential execution of iterative methods.
Yet, all our techniques extend to parallel implementation based on the message
passing paradigm (with using, e.g., MPI).
In an implementation of SpMxV in such a setting, the processing elements (or
processors) hold a part of the matrix and the input vector, and hold a part of
the output vector at the end.
A recent exposition of different algorithms can be found
elsewhere~\cite{kauc:14}. Typically, the processors perform scalar multiply operations on
the local matrix and the input vector elements, when all
required vector elements have been received from other processors.
The implementations of the MPI standard guarantees correct message delivery,
i.e., checksums are incorporated into the message so as to prevent transmission
errors (the receives can be done in-place and hence are protected).
However, the receiver will obviously get corrupted data if the sender sends
corrupted data.
Silent error can indeed strike at a given processor during local scalar multiply
operations.
Performing error detection and correction locally implies global error detection
and correction for the SpMxV. Note that, in this case, the local matrix elements
can form a matrix which cannot be assumed to be square in general (for some
iterative solvers they can be).
Furthermore, the mean time between failures (MTBF) reduces linearly with the
number of processors.
This is well-known for memoryless distributions of fault inter-arrival times and
remains true for arbitrary continuous distributions of finite
mean~\cite{aupy2014checkpointing}.
Therefore, resilient local matrix vector multiplies are required for resiliency in 
a parallel setting.

The rest of the paper is organized as follows.
Section~\ref{sec.related} provides an overview of related work.
Section~\ref{sec.CG} provides background on ABFT techniques for the PCG
algorithm, and presents both the \abftdetect and \abftcorrect approaches.
Section~\ref{sec.maths} is devoted to the abstract performance model.
Section~\ref{sec.experiments} reports numerical simulations comparing the
performance of \abftdetect, \abftcorrect and  \onlinechen.
Finally, we outline main conclusions and directions for future work in
Section~\ref{sec.conclusion}.

\section{Related work}
\label{sec.related}
We classify related work along the following topics: silent errors in general,
verification mechanisms for iterative methods, and ABFT techniques.

\subsection{Silent errors}
Considerable efforts have been directed at error-checking to reveal silent
errors.
Error detection is usually very costly. Hardware mechanisms, such as ECC memory,
can detect and even correct a fraction of errors, but in practice they are
complemented with software techniques. The simplest technique is triple modular
redundancy and voting~\cite{Lyons1962}, which induces a costly
verification.
For high-performance scientific applications, process replication (each process
is equipped with a replica, and messages are quadruplicated) is proposed in the
RedMPI library~\cite{Fiala2012}. Elliot et al.~\cite{Elliott2012} combine
partial redundancy and checkpointing, and confirm the benefit of dual and triple
redundancy.
The drawback is that twice the number of processing resources is required (for
dual redundancy).
A comprehensive list of general-purpose techniques and references is provided by
Lu et al.~\cite{LuZhengChien2013}.

Application-specific information can be very useful to enable ad-hoc solutions,
which dramatically decrease the cost of detection.
Many techniques have been advocated. They
include memory scrubbing~\cite{Hwang2012} and ABFT
techniques (see below).

As already stated, most papers assume on a selective reliability
setting~\cite{selective-reliability, selective-reliability-paper,
  selective-reliability-preprint, Sao2013}.
It essentially means that one can choose to perform any operation in reliable or
unreliable mode, assuming the former to be error-free but energy consuming and
the latter to be error-prone but preferable from an energy consumption point of view.

\subsection{Iterative methods}

Iterative methods offer a wide range of ad-hoc approaches.
For instance, instead of duplicating the computation, Benson et al.~\cite{BensonSS13}~suggest
coupling a higher-order with a lower-order scheme for PDEs.
Their method only detects an error but does not correct it.
Self-stabilizing corrections after error detection in the CG method are
investigated by Sao and Vuduc~\cite{Sao2013}.
Heroux and Hoemmen~\cite{Heroux2011} design a fault-tolerant GMRES capable of
converging despite silent errors.
Bronevetsky and de Supinski~\cite{Bronevetsky2008} provide a comparative study
of detection costs for iterative methods.

As already mentioned, a nice instantiation of the checkpoint and verification
mechanism that we study in this paper is provided by Chen~\cite{Chen2013},
who deals with sparse iterative solvers.
For PCG, the verification amounts to checking the orthogonality of two vectors
and to recomputing and checking the residual (see Section~\ref{sec.CG} for
further details).

As already mentioned, our abstract performance model is agnostic of the
underlying error-detection technique and takes the cost of verification as an
input parameter to the model.

\subsection{ABFT}

The very first idea of algorithm-based fault tolerance for linear algebra
kernels is given by Huang and Abraham~\cite{abft-matrix-operations}.
They describe an algorithm capable of detecting and correcting a single silent
error striking a matrix-matrix multiplication by means of row and column
checksums.
This germinal idea is then elaborated by Anfinson and
Luk~\cite{abft-linear-algebraic-model}, who propose a method to detect and
correct up to two errors in a matrix representation using just four column
checksums.
Despite its theoretical merit, the idea presented in their paper is actually
applicable only to relatively small matrices, and is hence out of our scope.
Bosilca et al.~\cite{Bosilca2009} and Du et al.~\cite{pengduppopp12} 
present two relatively recent survey.

The problem of algorithm-based fault-tolerance for sparse matrices is
investigated by Shantharam et al.~\cite{Shantharam2012}, who suggest a way to
detect a single error in an SpMxV at the cost of a few additional
\mbox{dot products}.
Sloan et al.~\cite{algorithmic-approach} suggest that this approach can be
relaxed using randomization schemes, and propose several checksumming techniques
for sparse matrices. These techniques are less effective than the previous ones,
not being able to protect the computation from faults striking the memory, but
provide an interesting theoretical insight.

\section{CG-ABFT}
\label{sec.CG}

We streamline our discussion on the CG method, however, the techniques that we
describe are applicable to any iterative solver that use sparse matrix vector
multiplies and vector operations.
This list includes many of the non-stationary iterative solvers such as CGNE,
BiCG, BiCGstab where sparse matrix transpose vector multiply operations also
take place.
In particular, we consider a PCG variant where the application of the
preconditioner reduces to the computation of two SpMxV with triangular
matrices~\cite{Benzi1998}, which are a sparse factorization of an approximate
inverse of the coefficient matrix.
In fact, the model discussed in this paper can be profitably employed for any
sparse inverse preconditioner that can be applied by means of one or more SpMxV.

We first provide a background on the CG method and
overview both Chen's stability tests~\cite{Chen2013} and our ABFT protection
schemes.

\begin{algorithm}
\begin{small}
  \begin{algorithmic}[1]
    \Require{$\vtr{A},\vtr{M} \in \R{n \times n}, \vtr{b}, \vtr{x}_0 \in \R{n}, \eps \in \R{}$}
    \Ensure{$\vtr{x} \in \R{n} \;:\; \| \vtr{A}\vtr{x} - \vtr{b} \| \leq \eps$}

    \State{$\vtr{r}_0 \leftarrow \vtr{b} - \vtr{A}\vtr{x}_0$};
    \State{$\vtr{z}_0 \leftarrow \trans{\vtr{M}}\vtr{M}\vtr{r}_0$};
    \State{$\vtr{p}_0 \leftarrow \vtr{z}_0$};
    \State{$i \leftarrow 0$};

    \While{$\norm{\vtr{r}_i} > \eps \left( \norm{\vtr{A}} \cdot \norm{\vtr{r}_0} + \norm{b}\right)$}
    \State{$\vtr{q} \leftarrow \vtr{A}\vtr{p}_i$};
    \State{$\alpha_i \leftarrow
      \norm{\vtr{r}_i}^2 / \trans{\vtr{p}}_i\vtr{q}$};
    \State{$\vtr{x}_{i + 1} \leftarrow \vtr{x}_i + \alpha\,\vtr{p}_i$};
    \State{$\vtr{r}_{i+1} \leftarrow \vtr{r}_i - \alpha\,\vtr{q}$};
    \State{$\vtr{z}_{i+1} \leftarrow \trans{\vtr{M}}\vtr{M}\vtr{r}_{i+1}$};
    \State{$\beta \leftarrow
      \norm{\vtr{r}_{i + 1}}^2 / \norm{\vtr{r}_{i}}^2$};
    \State{$\vtr{p}_{i + 1} \leftarrow \vtr{z}_{i + 1} + \beta\,\vtr{p}_i$};
    \State{$i \leftarrow i + 1$};
    \EndWhile
    \Return{$\vtr{x}_i$};
  \end{algorithmic}
  \caption{The PCG algorithm.}
  \label{alg:cg}
\end{small}
\end{algorithm}

%\subsection{CG and verification mechanisms}
%\label{sec.CG.chen}

The code for the variant of the PCG method we use is shown in
Algorithm~\ref{alg:cg}.
The main loop features three sparse matrix-vector
multiply, two inner products (for ${\trans{\vtr{p}}_i\vtr{q}}$ and
$\norm{\vtr{r}_{i + 1}}^2$), and three vector operations of the form $axpy$.

Chen's stability tests~\cite{Chen2013} amount to checking the orthogonality of
vectors $\vtr{p}_{i+1}$ and $\vtr{q}$, at the price of computing
$\frac{\trans{\vtr{p}}_{i+1}\vtr{q}}{\norm{\vtr{p}_{i+1}} \norm{\vtr{q}_{i}}}$,
and to checking the residual
at the price of an additional SpMxV operation $\vtr{A}\vtr{x}_{i}-\vtr{b}$.
The dominant cost of these verifications is the additional SpMxV operation.

Our only modification to Chen's original approach is that we also save the
sparse matrix $\vtr{A}$ in addition to the current iteration vectors.
This is needed when a silent error is detected: if this error comes for a
corruption in data memory, we need to recover with a valid copy of the data
matrix $\vtr{A}$.
This holds for the three methods under study, \onlinechen, \abftdetect and
\abftcorrect, which have exactly the same checkpoint cost.

We now give an overview of our own protection and verification mechanisms.
We use ABFT techniques to protect the SpMxV,  its computations (hence the vector
$\vtr{q}$), the matrix $\vtr{A}$ and the input vector $\vtr{p}_i$.
Since ABFT protection for vector operations is as costly as repeated computation,
we use triple modular redundancy (TMR) for them for simplicity.

Although theoretically possible, constructing ABFT mechanism to detect up to $k$
errors is practically not feasible for $k>2$.
The same mechanism can be used to correct up to $\lfloor k/2\rfloor$.
Therefore, we focus on detecting up to two errors and correcting the error if
there was only one.
That is, we detect up to two errors in the computation $\vtr{q}\leftarrow \vtr{A}\vtr{p}_i$
(two entries in $\vtr{q}$ are faulty), or in $\vtr{p}_i$, or in the sparse
representation of the matrix $\vtr{A}$.
With TMR, we assume that the errors in the computation are not overly frequent
so that two out of three are correct (we assume errors do not strike the vector
data here).
Our fault-tolerant PCG versions thus have the following ingredients: ABFT to
detect up to two errors in the SpMxV and correct the error, if there was only
one; TMR for vector operations; and checkpoint and roll-back in case errors are
not correctable.

%In the rest of this section,
We assume the selective reliability model in which all checksums and
checksum related operations are non-faulty, also the tests for the orthogonality
checks are non-faulty.

\section{ABFT-SpMxV}
\label{s:abftspmxv}
Here, we discuss the proposed ABFT method for the SpMxV
(combining ABFT with checkpointing is described later in Section~\ref{sec.model.CG}).
The proposed ABFT mechanisms are described for detecting single errors (Section~\ref{sec.single.err}), 
multiple errors (Section~\ref{sec.multiple.err}), and correcting a single error (Section~\ref{sec.err.correct}).

\subsection{Single error detection\label{sec.single.err}}

The overhead of the standard single error correcting ABFT technique is too high
for the sparse matrix-vector product case.
Shantharam et al.~\cite{Shantharam2012} propose a cheaper ABFT-SpMxV algorithm
that guarantees the detection of a single error striking either the computation
or the memory representation of the two input operands (matrix and vector).
Because their results depend on the sparse storage format adopted, throughout the
paper we will assume that sparse matrices are stored in the compressed storage
format by rows (CSR), that is by means of three distinct arrays, namely
$\vtxt{Colid} \in \N{\nnz(\vtr{A})}$, $\vtxt{Val} \in \R{\nnz(\vtr{A})}$ and
$\vtxt{Rowidx} \in \N{n + 1}$~\cite[Sec. 3.4]{saad}). Here $\nnz(\vtr{A})$
is the number of non-zero entries in $\vtr{A}$.

Shantharam et al. can protect $\vtr{y} \leftarrow \vtr{A} \vtr{x}$,
where $\vtr{A} \in \R {n \times n}$ and $\vtr{x,y} \in \R{n}$.
To perform error detection, they rely on a column checksum vector $\vtr{c}$
defined by
\begin{equation}
  c_j = \sum_{i = 0}^{n} a_{i, j}
  \label{eq:cj}
\end{equation}
and an auxiliary copy $\vtr{x}^\prime$ of the $\vtr{x}$ vector.
After having performed the actual SpMxV, to validate the result, it suffices to
compute $\sum_{i=1}^{n} y_i$, $\trans{\vtr{c}}\vtr{x}$ and
$\trans{\vtr{c}}\vtr{x}^\prime$, and to compare their values.
It can be shown~\cite{Shantharam2012} that in case of no errors, these three
quantities carry the same value, whereas if a single error strikes either the
memory or the computation, one of them must differ from the other two.
Nevertheless, this method requires $\vtr{A}$ to be strictly diagonally dominant,
a condition that seems to restrict too much the practical applicability of their
method.
Shantharam et al.~need this condition to ensure detection of errors striking
an entry of $\vtr{x}$ corresponding to a zero checksum column of $\vtr{A}$.
We further analyze that case and show how to overcome the issue without imposing
any restriction on $\vtr{A}$.

A nice way to characterize the problem is expressing it in geometrical terms.
Consider the computation of a single entry of the checksum as
\begin{equation*}
  (\trans{\vtr{w}}\vtr{A})_j = \sum_{i=1}^{n} w_i a_{i,j} =
  \trans{\vtr{w}}\vtr{A}^j,
\end{equation*}
where $\vtr{w} \in \R{n}$ denotes the weight vector and $\vtr{A}^j$ the $j$-th
column of $\vtr{A}$.
Let us now interpret such an operation as the result of the scalar product
$
\left\langle \cdot, \cdot \right\rangle : \R{n} \times \R{n} \rightarrow \R{}
$
defined by
$
\left\langle \vtr{u}, \vtr{v} \right\rangle \mapsto \trans{\vtr{u}}\vtr{v}.
$
It is clear that a checksum entry is zero if and only if the corresponding
column of the matrix is orthogonal to the weight vector.
In~\eqref{eq:cj}, we have chosen $\vtr{w}$ to be such that $w_i = 1$ for
$1 \leq i \leq n$, in order to make the computation easier.
Let us see now what happens without this restriction.

The problem reduces to finding a vector $\vtr{w} \in \R{n}$ that is not
orthogonal to any vector out of a basis
$\mathcal{B} = \left\lbrace \vtr{b}_1, \dots, \vtr{b}_n \right\rbrace$ of $\R{n}$
-- the rows of the input matrix.
Each of these $n$ vectors is perpendicular to a hyperplane $h_i$ of $\R{n}$,
and $\vtr{w}$ does not verify the condition
\begin{equation}
  \label{eq:orth-condition}
    \left\langle \vtr{w}, \vtr{b}_i \right\rangle \neq 0,
\end{equation}
for any $i$, if and only if it lies on $h_i$. Since the Lebesgue measure in $\R{n}$
of an hyperplane of $\R{n}$ itself is zero, the union of these hyperplanes is
measurable and has measure 0.
Therefore, the probability that a vector $\vtr{w}$ randomly picked in $\R{n}$
does not satisfy condition~\eqref{eq:orth-condition} for any $i$ is zero.

Nevertheless, there are many reasons to consider zero checksum columns.
First of all, when working with finite precision, the number of elements in
$\R{n}$ one can have is finite, and the probability of randomly picking a vector
that is orthogonal to a given one could be larger than zero.
Moreover, a coefficient matrix usually comes from the discretization of a
physical problem, and the distribution of its columns cannot be considered as
random.
Finally, using a randomly chosen vector instead of $\trans{(1, \dots, 1)}$
increases the number of required floating point operations, causing a growth of
both the execution time and the number of rounding errors
(see Section~\ref{sec.experiments}).
Therefore, we would like to keep \mbox{$\vtr{w}=\trans{(1, \dots, 1)}$} as the
vector of choice, in which case we need to protect SpMxV with matrices having
zero column sums.
There are many matrices with this property, for example the Laplacian matrices
of graphs~\cite[Ch.~1]{Chung1997}.

\begin{algorithm}
% \begin{spacing}{1.}
\begin{small}
\begin{algorithmic}[1]
\Require{$\vtr{A} \in \R {n \times n}$,  $\vtr{x} \in \R{n}$}
\Ensure{$\vtr{y} \in \R{n}$ such that $\vtr{y} = \vtr{A} \vtr{x}$ or the detection of a single error}

\State{$\vtr{x}^\prime \leftarrow \vtr{x}$};
\State{$[\vtr{\underline{w}}, \vtr{c}, k, c_{r}] =$} \Call{computeChecksum}{$\vtr{A}$};
\Return{\Call{SpMxV}{$\vtr{A}$, $\vtr{x}$, $\vtr{x}^\prime$, $\vtr{\underline{w}}$, $\vtr{c}$, $k$, $c_{r}$}};

\Statex

\Function{computeChecksum}{$\vtr{A}$}
	\State{Generate $\vtr{\underline{w}} \in \R{n+1}$};
	\State{$\vtr{w} \leftarrow \vtr{\underline{w}}_{1:n}$};
	\State{$\vtr{c} \leftarrow \trans{\vtr{w}} \vtr{A}$};
	\If{$\min(\mid \vtr{c} \mid) = 0$};
		\State{Find $k$ that verifies~\eqref{eq:shift-condition}};
		\State {$\vtr{c} \leftarrow \vtr{c} + k$};
	\EndIf
	\State{$c_{\vtxt{r}} \leftarrow \trans{\vtr{\underline{w}}}\vtxt{Rowindex} $};
	\Return{$\vtr{\underline{w}}, \vtr{c}, k, c_{r}$};
\EndFunction

\Statex

\Function{SpMxV}{$\vtr{A}$, $\vtr{x}$, $\vtr{x}^\prime$, $\vtr{\underline{w}}$, $\vtr{c}$, $k$, $c_{r}$}
	\State{$\vtr{w} \leftarrow \vtr{\underline{w}}_{1:n}$};
	\State{$s_r \leftarrow 0$};
	\For {$i\leftarrow1$ to $n$}
		\State{$y_i \leftarrow 0$};
		\State{$s_r \leftarrow s_r + \vtxt{Rowindex}_i$};
		\For {$j\leftarrow\vtxt{Rowindex}_i$ to $\vtxt{Rowindex}_{i+1} - 1$}
			\State{$ind \leftarrow \vtxt{Colid}_j$};
			\State{$y_i \leftarrow y_i + \vtxt{Val}_j \cdot x_{ind}$};
		\EndFor
	\EndFor
	\State{$y_{n+1} \leftarrow k\; \trans{\vtr{w}}\vtr{x}^\prime$};
	\State{$c_{y} \leftarrow  \trans{\vtr{\underline{w}}} \vtr{y}$};
        \State{$d_{x} \leftarrow  \trans{\vtr{c}} \vtr{x}$};
        \State{$d_{x^\prime} \leftarrow  \trans{\vtr{c}} \vtr{x}^\prime$};
        \State{$d_{r} \leftarrow c_r - s_r$};
	\If{$d_x = 0 \wedge d_{x^\prime} = 0 \wedge d_r = 0$}
		\Return{$\vtr{y}_{1:n}$};
	\Else
		\State{\textbf{error} (``Soft error is detected")};
	\EndIf
\EndFunction
\end{algorithmic}
\caption{Shifting checksum algorithm.}
\label{alg:ssr-shift}
% \end{spacing}
\end{small}
\end{algorithm}

In Algorithm~\ref{alg:ssr-shift}, we propose an ABFT SpMxV method that uses
weighted checksums and does not require the matrix to be strictly diagonally
dominant.
The idea is to compute the checksum vector and then shift it by adding to all
entries a constant value chosen so that all elements of the new
vector are different from zero.
We give the generalized result in Theorem~\ref{th:ssr-shift}.

\begin{theorem}[Correctness of Algorithm~\ref{alg:ssr-shift}]
\label{th:ssr-shift}
Let $\vtr{A} \in \R{n \times n}$ be a square matrix, let
$\vtr{x}, \vtr{y} \in \R{n}$ be the input and output vector respectively, and
let $\vtr{x}^\prime = \vtr{x}$.
Let us assume that the algorithm performs the computation
\begin{equation}
\label{eq:encoded-spmxv-mod}
\widetilde{\vtr{y}} \leftarrow \widetilde{\vtr{A}} \widetilde{\vtr{x}},
\end{equation}
where $\widetilde{\vtr{A}} \in \R{n \times n}$ and $\widetilde{\vtr{x}} \in \R{n}$
are the possibly faulty representations of $\vtr{A}$ and $\vtr{x}$ respectively,
while $\widetilde{\vtr{y}} \in \R{n}$ is the possibly erroneous result of the
sparse matrix-vector product. Let us also assume that the encoding scheme relies
on
\begin{enumerate}
\item an auxiliary checksum vector
\begin{equation*}
  \vtr{c} = \left[\sum_{i=1}^{n}a_{i,1}+k,\ldots,\sum_{i=1}^{n}a_{i,n}+k\right],
\end{equation*}
where $k$ is such that
\begin{equation}
\label{eq:shift-condition}
c_j=\sum_{i=1}^{n}a_{i,j} + k \neq 0,
\end{equation}
 for $1 \le j \le n$,
\item an auxiliary checksum $y_{n+1} = k \sum_{i=i}^{n}\widetilde{x}_i$,
\item an auxiliary counter $s_r$ initialized to 0 and updated at
runtime by adding the value of the hit element each time the $\vtxt{Rowindex}$
array is accessed (line 20 of Algorithm~\ref{alg:ssr-shift}),
\item an auxiliary checksum $c_{r}=\sum_{i=1}^{n}\vtxt{Rowindex}_i\in\N{}$.
\end{enumerate}
Then, a single error in the computation of the SpMxV causes one of the following
conditions to fail:
\begin{enumerate}[i.]
\item $\trans{\vtr{c}}\widetilde{\vtr{x}} = \sum_{i=1}^{n+1} \widetilde{y}_i$,
\label{it:ssr-shift-i}
\item $\trans{\vtr{c}}\vtr{x}^\prime = \sum_{i=1}^{n+1} \widetilde{y}_i$,
\label{it:ssr-shift-ii}
\item $s_r$ = $c_{r}$.
\label{it:ssr-shift-iii}
\end{enumerate}
\end{theorem}

\begin{proof}
We will consider three possible cases, namely
\begin{enumerate}[a.]
\item a faulty arithmetic operation during the computation of $\vtr{y}$,
\label{it:ssr-shift-a}
\item a bit flip in the sparse representation of $\vtr{A}$,
\label{it:ssr-shift-b}
\item a bit flip in an element of of $\vtr{x}$.
\label{it:ssr-shift-c}
\end{enumerate}

\textbf{Case~\ref{it:ssr-shift-a}.} Let us assume, without loss of generality,
that the error has struck at the $p$th position of $\vtr{y}$, which
implies $\widetilde{y}_i = y_i$ for $1 \leq i \leq n$ with $i \neq p$ and
$\widetilde{y}_p = y_p + \eps$, where
$\eps \in \R{} \setminus \left\lbrace 0 \right\rbrace$ represents the value of
the error that has occurred.
Summing up the elements of $\widetilde{\vtr{y}}$ gives
\begin{eqnarray}
\sum_{i=1}^{n+1} \widetilde{y}_i
%&=& \sum_{i=1}^{n+1} y_i + \eps \nonumber \\
&=& \sum_{i=1}^{n} \sum_{j=1}^{n} a_{i,j} \widetilde{x}_j +
	k \sum_{j=1}^{n} \widetilde{x}_j + \eps \nonumber \nonumber \\
%&=& \sum_{j=1}^{n} \sum_{i=1}^{n} a_{i,j} \widetilde{x}_j +
%	k \sum_{j=1}^{n} \widetilde{x}_j + \eps \nonumber \nonumber \\
%&=& \sum_{j=1}^{n} \left(\sum_{i=1}^{n} a_{i,j} +
%	k \right) \widetilde{x}_j + \eps \nonumber \nonumber \\
&=& \sum_{j=1}^{n} c_j \widetilde{x}_j + \eps \nonumber \nonumber \\
%&=& \trans{\vtr{c}}\widetilde{\vtr{x}} + \eps \nonumber \nonumber \\
&=& \trans{\vtr{c}}\widetilde{\vtr{x}} + \eps, \nonumber \nonumber
\end{eqnarray}
that violates condition~\eqref{it:ssr-shift-i}.

\textbf{Case~\ref{it:ssr-shift-b}.} A single error in the $\vtr{A}$ matrix can
strike one of the three vectors that constitute its sparse representation:
\begin{itemize}

\item a fault in $\vtxt{Val}$ that alters the value of an element $a_{i,j}$
implies an error in the computation of $\widetilde{y}_{i}$, which leads to the
violation of the safety condition~\eqref{it:ssr-shift-i} because
of~\eqref{it:ssr-shift-a},

\item a variation in $\vtxt{Colid}$ can zero out an element in position
$a_{i,j}$ shifting its value in position $a_{i,j^\prime}$, leading again to an
erroneous computation of $\widetilde{y}_{i}$,

\item a transient fault in $\vtxt{Rowindex}$ entails an incorrect value of
$s_r$ and hence a violation of condition \eqref{it:ssr-shift-iii}.
\end{itemize}

\textbf{Case~\ref{it:ssr-shift-c}.} Let us assume, without loss of generality,
an error in position $p$ of $\vtr{x}$. Hence we have that
$\widetilde{x}_i = x_i$ for $1 \leq i \leq n$ with $i \neq p$ and
$\widetilde{x}_p = x_p + \eps$, for some $\eps\in\R{}\setminus\lbrace 0\rbrace$.
Noting that $\vtr{x} = \vtr{x}^\prime$, the sum of the elements of
$\widetilde{\vtr{y}}$ gives
\begin{eqnarray}
\sum_{i=1}^{n+1} \widetilde{y}_i
&=& \sum_{i=1}^{n} \sum_{j=1}^{n} a_{i,j}\widetilde{x}_j + k \sum_{j=1}^{n}\widetilde{x}_j \nonumber \\
&=& \sum_{i=1}^{n} \sum_{j=1}^{n} a_{i,j}x_j + k \sum_{j=1}^{n}x_j +
\eps \sum_{i=1}^{n} a_{i,p} + \eps k \nonumber \\
%&=& \sum_{j=1}^{n} \sum_{i=1}^{n} a_{i,j}x_j + k \sum_{j=1}^{n}x_j +
%\eps \sum_{i=1}^{n} a_{i,p} + \eps k \nonumber \\
%&=& \sum_{j=1}^{n} \left(\sum_{i=1}^{n} a_{i,j} + k \right) x_j +
%\eps\left( \sum_{i=1}^{n} a_{i,p} + k \right) \nonumber \\
&=& \sum_{j=1}^{n} c_j x_j +
\eps\left( \sum_{i=1}^{n} a_{i,p} + k \right) \nonumber \\
&=& \trans{\vtr{c}} \vtr{x}^\prime +
\eps\left( \sum_{i=1}^{n} a_{i,p} + k \right), \nonumber
\end{eqnarray}
that violates \eqref{it:ssr-shift-ii} since $\sum_{i=1}^{n} a_{i,p} + k \neq 0$
by definition of $k$.
\end{proof}

Let us remark that $\textsc{computeChecksum}$ in Algorithm~\ref{alg:ssr-shift}
does not require the input vector $\vtr{x}$ of SpMxV as an argument. Therefore,
in the common scenario of many SpMxV with the same matrix, it is enough to invoke
it once to protect several \mbox{matrix-vector} multiplications.
This observation will be crucial when discussing the performance of these
checksumming techniques.

Shifting the sum checksum vector by an amount
is probably the simplest deterministic  approach to relax the strictly
diagonal dominance hypothesis, but it is not the only one.
An alternative solution is described in Algorithm~\ref{alg:ssr-split},
which basically exploits the distributive property of matrix multiplication over
matrix addition.
The idea is to split the original matrix $\vtr{A}$ into two matrices of the same
size, $\vtr{A}$ and $\widehat{\vtr{A}}$, such that no column of either matrix
has a zero checksum.
Two standard ABFT multiplications, namely $\vtr{y} \leftarrow \vtr{A} \vtr{x}$
and $\widehat{\vtr{y}} \leftarrow \vtr{A} \widehat{\vtr{x}}$, are then
performed.
If no error occurs neither in the first nor in the second computation, the sum
of $\vtr{y}$ and $\widehat{\vtr{y}}$ is computed in reliable mode and then
returned.
Let us note that, as we expect the number of non-zeros of $\widehat{\vtr{A}}$ to
be much smaller than $n$, we store sparsely both the checksum vector of
$\widehat{\vtr{A}}$ and the $\widehat{\vtr{y}}$ vector.

We do not write down an extended proof of the correctness of this algorithm, and
limit ourselves to a short sketch.
We consider the same three cases as in the proof of
Theorem~\ref{th:ssr-shift}, without introducing any new idea.
An error in the computation of $\vtr{y}$ or $\widehat{\vtr{y}}$ can be detected
using the \mbox{dot product} between the corresponding column checksum and the
$\vtr{x}$ error. An error in $\vtr{A}$ can be detected by either $c_{r}$
or an erroneous entry in $\vtr{y}$ or $\widehat{\vtr{y}}$, as the matrix loop
structure of the sparse multiplication algorithm has not been changed.
Finally, an error in the $p$th component of $\vtr{x}$ would make the sum of the
entries of $\vtr{y}$ and $\widehat{\vtr{y}}$ differ from
$\trans{\vtr{c}}\vtr{x}^\prime$ and $\trans{\widehat{\vtr{c}}}\vtr{x}^\prime$,
respectively.

\begin{algorithm}
% \begin{spacing}{1.2}
\begin{small}
\begin{algorithmic}[1]
\Require{$\vtr{A} \in \R {n \times n}$,
  $\vtr{x} \in \R{n}$}
\Ensure{$\vtr{y} \in \R{n}$ such that $\vtr{y} = \vtr{A} \vtr{x}$
  or the detection of a single error}

\State{$[\vtr{c}, k, c_{r}] =$} \Call{computeChecksum}{$\vtr{A}$};
\Return{\Call{SpMxV}{$\vtr{A}$, $\vtr{x}$, $\vtr{c}$, $\widehat{\vtr{c}}$, $\vtr{b}$, $c_{r}$}};

\Statex

\Function{computeChecksum}{$\vtxt{A}$}
	\State{$\vtr{c}, \vtr{m} \leftarrow 0$};
	\For {$i\leftarrow1$ to $\nnz(\vtr{A})$}
		\State{$ind \leftarrow \vtxt{Colid}_i$};
		\State{$c_{ind} \leftarrow c_{ind} + \vtxt{Val}_i$};
		\State{$m_{ind} \leftarrow i$};
	\EndFor
	\State{$k \leftarrow 0$};
	\For{$i \leftarrow 1$ to $n$}
		\If{$c_i = 0 \wedge m_i \neq 0$}
			\State{$b_{m_i} \leftarrow \textbf{true}$};
			\State{$\widehat{c}_i \leftarrow \vtxt{Val}_{m_i}$};
			\State{$c_i \leftarrow c_i - \widehat{c}_i$};
		\EndIf
	\EndFor
	\State{$c_r \leftarrow \sum_{i=1}^{n} \vtxt{Rowindex}_i $};
	\Return{$\vtr{c}, \widehat{\vtr{c}}, \vtr{b}, c_r$};
\EndFunction

\Statex

\Function{SpMxV}{$\vtr{A}$, $\vtr{x}$, $\vtr{c}$, $\widehat{\vtr{c}}$, $\vtr{b}$, $c_r$}
	\State{$\vtr{x}^\prime \leftarrow \vtr{x}$};
	\For {$i\leftarrow1$ to $n$}
		\State{$y_i \leftarrow 0$};
	\EndFor
	\State{$s_r \leftarrow 0$};
	\For {$i\leftarrow1$ to $n$}
		\State{$s_r \leftarrow s_r + \vtxt{Rowindex}_i$};
		\For {$j\leftarrow\vtxt{Rowindex}_i$ to $\vtxt{Rowindex}_{i+1} - 1$}
			\State{$ind \leftarrow \vtxt{Colid}_j$};
			\If{$b_j$}
				\State{$\widehat{y}_i \leftarrow \widehat{y}_i + \vtxt{Val}_j \cdot x_{ind}$};
			\Else
				\State{$y_i \leftarrow y_i + \vtxt{Val}_j \cdot x_{ind}$};
			\EndIf
		\EndFor
	\EndFor

	\State{$c_{y} \leftarrow \sum_{i=1}^{n} y_i$;
          $c_{\widehat{y}} \leftarrow \sum_{i=1}^{n} \widehat{y}_i$};
        \State{$d_{x} \leftarrow \trans{\vtr{c}}\vtr{x} - c_y$;
          $d_{\widehat{x}} \leftarrow \trans{\vtr{\widehat{c}}}\vtr{x} - c_{\widehat{y}}$};
        \State{$d_{x^\prime} \leftarrow \trans{\vtr{c}}\vtr{x}^\prime - c_y$;
          $d_{\widehat{x}^\prime} \leftarrow \trans{\vtr{\widetilde{c}}}\vtr{x}^\prime - c_{\widehat{y}}$};
        \State{$d_{r} \leftarrow c_r - s_r $};
	\If{$ d_x = 0 \wedge d_{x^\prime} = 0 \wedge d_{\widetilde{x}} = 0 \wedge d_{\widetilde{x}^\prime} = 0  \wedge d_r = 0$}
		\Return{$\vtr{y} + \widehat{\vtr{y}}$};
	\Else
		\State{\textbf{error} (``Soft error is detected")};
	\EndIf
\EndFunction
\end{algorithmic}
\caption{Splitting checksum algorithm.}
\label{alg:ssr-split}
% \end{spacing}
\end{small}
\end{algorithm}

\begin{table}
\centering
\caption{Overhead comparison for Algorithm~\ref{alg:ssr-shift} and Algorithm~\ref{alg:ssr-split}. Here $n$ denotes the size of the matrix and $\np$ the number of null sum columns.}
\small{
\begin{tabular}{l|c|c}
									& Algorithm~\ref{alg:ssr-shift} 
											& Algorithm~\ref{alg:ssr-split} \\
\midrule
initialization of $\vtr{y}$ 		& $n$ 		& $n$ \\
%computation of $\vtr{y}$ and $\widehat{\vtr{y}}$
%									& $\nnz$	& $\nnz$ \\
computation of $y_{n+1}$			& $n$ 		& - \\
SpMxV overhead		& -			& $\np$ \\
checksum check	& $2n$		& $2n + 2\np$ \\
computation of $c_{y}$ and $c_{\widehat{y}}$
									& $n$ 		& $n + \np$ \\
computation of $\vtr{y} + \widehat{\vtr{y}}$
									& -			& $\np$ \\
\midrule
Total SpMxV overhead	& $5n$		& $4n + 5\np$ \\
\bottomrule
\end{tabular}
}
\label{tab:shift-split-comparison}
\end{table}

The evaluation of the performance of the two algorithms, though straightforward
from the point of view of the computational cost, has to be carefully assessed
in order to devise a valid and practical trade-off between
Algorithm~\ref{alg:ssr-shift} and Algorithm~\ref{alg:ssr-split}.

In both cases \textsc{computeChecksum} introduces an overhead of
$\mO{\nnz(\vtr{A})}$, but the shift version should in general be faster
containing less assignments than its counterpart, and this changes the
constant factor hidden by the asymptotic notation.
Nevertheless, as we are interested in performing many SpMxV with a same matrix,
this pre-processing overhead becomes negligible.

The function \textsc{SpMxV} has to be invoked once for each multiplication, and
hence more care is needed.
Copying $\vtr{x}$ and initializing $\vtr{y}$ both require $n$ operations, and
the multiplication is performed in time $\mO{\nnz(A)}$, but the split version
pays an $\np$ more to read the values of the sparse vector $\vtr{b}$.
The cost of the verification step depends instead on the number of zeroes of the
original checksum vector, that is also the number of non-zero elements of the
sparse vector $\widehat{\vtr{c}}$. Let us call this quantity $\np$.
Then the overhead is $4n$ for the shifting and $3n + 3\np$ for the splitting,
that requires also the sum of two sparse vectors to return the result.
Hence, as summarized in Table~\ref{tab:shift-split-comparison}, the two methods
bring different overhead into the computation.
Comparing them, it is immediate to see that the shifting method is cheaper when
\begin{equation*}
\np > \dfrac{n}{5},
\end{equation*}
while it has more operations to do when the opposite inequality holds.
For the equality case, we can just choose to use the first method because of the
cheaper preprocessing phase.
In view of this observation, it is possible to devise a simple algorithm that
exploits this trade-off to achieve better performance.
It suffices to compute the checksum vector of the input matrix, count the number
of non-zeros and choose which detection method to use accordingly.

We also note that by splitting the matrix $\vtr{A}$ into say $\ell$ pieces and checksumming 
each piece separately we can possibly protect $\vtr{A}$ from up to $\ell$ errors, 
by protecting each piece against a single error (obviously the multiple errors should hit different pieces).

\subsection{Multiple error detection}\label{sec.multiple.err}

With some effort, the shifting idea in Algorithm~\ref{alg:ssr-shift} can be
extended to detect errors striking a single SpMxV.
Let us consider the problem of detecting up to $k$ errors in the computation of
$\vtr{y}\gets \vtr{A} \vtr{x}$ introducing an overhead of $\mO{kn}$.
Let $k$ weight vectors $\vtr{w}^{(1)}, \dots, \vtr{w}^{(k)} \in \R{n}$ be such
that any sub-matrix of
\[
\vtr{W} =
\left[
\vtr{w}^{(1)}\;
\vtr{w}^{(2)}\;
\dots\;
\vtr{w}^{(k)}\;
\right]
\]
has full rank.
To  build our ABFT scheme let us note that, if no error occurs, for each weight
vector $\vtr{w}^{(\ell)}$ it holds that
\[
\vtr{w}^{{\trans{(\ell)}}} \vtr{A} = \left[
\sum_{i=1}^{n} w^{(\ell)}_i a^{\vphantom{(\ell)}}_{i,1},
\dots,
\sum_{i=1}^{n} w^{(\ell)}_i a^{\vphantom{(\ell)}}_{i,n}
\right],
\]
and hence that
\begin{eqnarray}
\vtr{w}^{{\trans{(\ell)}}} \vtr{A} \vtr{x} &=&
\sum_{i=1}^{n} w^{(\ell)}_i a^{\vphantom{(\ell)}}_{i,1} x^{\vphantom{(\ell)}}_1 +
\dots +
\sum_{i=1}^{n} w^{(\ell)}_i a^{\vphantom{(\ell)}}_{i,n} x^{\vphantom{(\ell)}}_n \nonumber \\
&=&
\sum_{i=1}^{n} \sum_{j=1}^{n} w^{(\ell)}_i a^{\vphantom{(\ell)}}_{i,j} x^{\vphantom{(\ell)}}_j. \nonumber
\end{eqnarray}
Similarly, the sum of the entries of $\vtr{y}$ weighted with the same
$\vtr{w}^{(\ell)}$ is
\begin{eqnarray}
\sum_{i=1}^{n} w^{(\ell)}_i y_i
&=& w^{(\ell)}_1 y^{\vphantom{(\ell)}}_1 + \dots + w^{(\ell)}_n y^{\vphantom{(\ell)}}_n \nonumber \\
&=& w^{(\ell)}_1 \sum_{j=1}^{n} a_{1,j} x_j
	+ \dots + w^{(\ell)}_n \sum_{j=1}^{n} a^{\vphantom{(\ell)}}_{n,j} x^{\vphantom{(\ell)}}_j \nonumber \\
&=& \sum_{i=1}^{n} \sum_{j=1}^{n} w^{(\ell)}_i a^{\vphantom{(\ell)}}_{i,j} x_j, \nonumber
\end{eqnarray}
and we can conclude that
\[
\sum_{i=1}^{n} w^{(\ell)}_i y^{\vphantom{(\ell)}}_i =
\left(\vtr{w}^{\trans{(\ell)}} \vtr{A} \right) \vtr{x},
\]
for any $\vtr{w}^{(\ell)}$ with $1 \leq \ell \leq k$.

To convince ourself that with these checksums it is actually possible to detect
up to $k$ errors, let us suppose that $k^\prime$ errors, with $k^\prime \leq k$,
occur in positions $p_1, \dots, p_{k^\prime}$, and let us denote by
$\widetilde{\vtr{y}}$ the faulty vector where
 $\widetilde{y}_{p_i} = y_{p_i} + \eps_{p_i}$ for
 $\eps_{p_i} \in \R{} \setminus \lbrace 0 \rbrace$ and $1 \leq i \leq k^\prime$
and $\widetilde{y}_{i} = y_{i}$ otherwise.
Then for each weight vector we have
\[
\sum_{i=1}^{n} w^{(\ell)}_i \widetilde{y}^{\vphantom{(\ell)}}_i - \sum_{i=1}^{n} w^{(\ell)}_i y^{\vphantom{(\ell)}}_i = \sum_{j=1}^{k^\prime} w^{(\ell)}_{p_j} \eps^{\vphantom{(\ell)}}_{p_j}.
\]
Said otherwise, the occurrence of the $k^\prime$ errors is not detected if and
only if, for $1 \leq \ell \leq k$, all the $\eps_{p_i}$ respect
\begin{equation}
\sum_{j=1}^{k^\prime} w^{(\ell)}_{p_j} \eps^{\vphantom{(\ell)}}_{p_j} = 0\;.
\label{eq:kernel-condition}
\end{equation}
We claim that there cannot exist a vector
$\trans{(\eps_{p_1}, \dots, \eps_{p_k^\prime})} \in \R{k^\prime}
\setminus \left\lbrace 0 \right\rbrace$
such that all the conditions in \eqref{eq:kernel-condition} are
simultaneously verified.
These conditions can be expressed in a more compact way as a linear
system
\[
\left(
\begin{matrix}
w^{(1)}_{p_1} 			& \cdots 	& w^{(1)}_{p_{k^\prime}}	 	\\
\vdots					& \ddots	& \vdots 						\\
w^{(k)}_{p_1}	& \cdots	& w^{(k)}_{p_{k^\prime}}	\\
\end{matrix}
\right)
\left(
\begin{matrix}
\eps^{\vphantom{(\ell)}}_{p_1} 			\\
\vdots 				\\
\eps^{\vphantom{(\ell)}}_{p_{k^\prime}}	\\
\end{matrix}
\right)
=
\left(
\begin{matrix}
0		\\
\vdots	\\
0		\\
\end{matrix}
\right).
\]

Denoting by $\vtr{W}^*$ the coefficient matrix of this system, it is clear that
the errors cannot be detected if only if
\mbox{$\trans{(\eps_{p_1}, \dots \eps_{p_{k^\prime}})} \in \ker(\vtr{W}^*)
\setminus \left\lbrace 0 \right\rbrace$}.
Because of the properties of $\vtr{W}$, we have that $\rk(\vtr{W}^*) = k$.
Moreover, it is clear that the rank of the augmented matrix
\[
\left(
\begin{matrix}
w^{(1)}_{p_1} 			& \cdots 	& w^{(1)}_{p_{k^\prime}} \\
\vdots					& \ddots	& \vdots \\
w^{(k)}_{p_1}	& \cdots	& w^{(k)}_{p_{k^\prime}} \\
\end{matrix}
\, \middle\vert \,
\begin{matrix}
  0 \\
  \vdots \\
  0
\end{matrix}
\right)
\]
is $k$ as well.
Hence, by means of the Rouch\'{e}--Capelli theorem, the solution of the system is
unique and the null space of $\vtr{W}^*$ is trivial.
Therefore, this construction can detect the occurrence of $k^\prime$ errors
during the computation of $\vtr{y}$ by comparing the values of the weighted sums
$\trans{\vtr{y}}\vtr{w}^{(\ell)}$ with the result of the \mbox{dot product}
$(\vtr{w}^{\trans{(\ell)}}\vtr{A})\vtr{x}$, for $1 \leq \ell \leq k$.

However, to get a true extension of the algorithm described in the previous
section, we also need to make it able to detect errors that strike the sparse
representation of $\vtr{A}$ and that of $\vtr{x}$.
The first case is simple, as the $k$ errors can strike the $\vtxt{Val}$ or
$\vtxt{Colid}$ arrays, leading to at most $k$ errors in $\vtr{\widetilde{y}}$,
or in $\vtxt{Rowindex}$, where they can be caught using $k$ weighted checksums
of the $\vtxt{Rowindex}$ vector.

Detection in $\vtr{x}$ is much trickier, since neither the algorithm just
described nor a direct generalization of Algorithm~\ref{alg:ssr-shift} can
manage this case. Nevertheless, a proper extension of the shifting technique is
still possible.
Let us note that there exists a matrix $\vtr{M} \in \R {k \times n}$ such that
\[
\trans{\vtr{W}}\vtr{A} + \vtr{M} = \vtr{W}.
\]
The elements of such an $\vtr{M}$ can be easily computed, once that the checksum
rows are known.
Let $\widetilde{\vtr{x}} \in \R{n}$ be the faulty vector, defined by
\[
\widetilde{x}_i =
\left\lbrace
\begin{array}{ll}
x_i + \eps_{p_i},\qquad\qquad & 1 \leq i \leq k', \\[11pt]
x_i & otherwise.
\end{array}
\right.
\]
for some $k^\prime \leq k$, and let us define
$\widetilde{\vtr{y}} = \vtr{A}\widetilde{\vtr{x}}$.
Now, let us consider a checksum vector $\vtr{x}^\prime \in \R {n}$
such that $\vtr{x}^\prime = \vtr{x}$ and let assume that it cannot be modified
by a transient error.
For $1 \leq \ell \leq k$, it holds that
\begin{footnotesize}
\begin{align*}
\sum_{i=1}^{n} w_i^{(\ell)} \widetilde{y}^{\vphantom{(\ell)}}_i
+ \sum_{j=1}^{n} m^{\vphantom{(\ell)}}_{\ell,j} \widetilde{x}^{\vphantom{(\ell)}}_j
&= \sum_{i=1}^{n}\sum_{j=1}^{n} w_i^{(\ell)} a^{\vphantom{(\ell)}}_{i,j} x^{\vphantom{(\ell)}}_j
+ \sum_{i=1}^{k^\prime}\eps^{\vphantom{(\ell)}}_{p_i}
\left(\sum_{j=1}^{n} w_j^{(\ell)} a^{\vphantom{(\ell)}}_{j,p_i} \right)
+\sum_{j=1}^{n} m^{\vphantom{(\ell)}}_{\ell,j} x^{\vphantom{(\ell)}}_j
+ \sum_{i=1}^{k^\prime}\eps^{\vphantom{(\ell)}}_{p_i} m^{\vphantom{(\ell)}}_{l,p_i} \nonumber \\
&= \sum_{j=1}^{n}\sum_{i=1}^{n} w_i^{(\ell)} a^{\vphantom{(\ell)}}_{i,j} x^{\vphantom{(\ell)}}_j
+ \sum_{j=1}^{n} m^{\vphantom{(\ell)}}_{\ell,j} x^{\vphantom{(\ell)}}_j
+ \sum_{i=1}^{k^\prime}\eps^{\vphantom{(\ell)}}_{p_i}
\left(\sum_{j=1}^{n} w_j^{(\ell)} a^{\vphantom{(\ell)}}_{j,p_i} + m^{\vphantom{(\ell)}}_{l,p_i}\right) \nonumber \\
&= \sum_{j=1}^{n}\left(\sum_{i=1}^{n} w_i^{(\ell)} a^{\vphantom{(\ell)}}_{i,j}\right) x^{\vphantom{(\ell)}}_j
+ \sum_{j=1}^{n} m^{\vphantom{(\ell)}}_{\ell,j} x^{\vphantom{(\ell)}}_j
+ \sum_{i=1}^{k^\prime}\eps^{\vphantom{(\ell)}}_{p_i} w^{(\ell)}_{p_i} \nonumber \\
&= \sum_{j=1}^{n}
\left(\sum_{i=1}^{n} w_i^{(\ell)} a^{\vphantom{(\ell)}}_{i,j} + m_{\ell,j}\right) x^{\vphantom{(\ell)}}_j
+ \sum_{i=1}^{k^\prime}\eps_{p_i} w^{(\ell)}_{p_i} \nonumber \\
&= \vtr{w}^{\trans{(\ell)}} \vtr{x}
+ \sum_{i=1}^{k^\prime}\eps^{\vphantom{(\ell)}}_{p_i} w^{(\ell)}_{p_i} \nonumber \\
&= \vtr{w}^{\trans{(\ell)}} \vtr{x}^\prime
+ \sum_{i=1}^{k^\prime}\eps^{\vphantom{(\ell)}}_{p_i} w^{(\ell)}_{p_i} \nonumber\;.
\end{align*}
\end{footnotesize}
Therefore, an error is not detected if and only if the linear system
\[
\left(
\begin{matrix}
w^{(1)}_{p_1} 			& \cdots 	& w^{(1)}_{p_{k^\prime}}	 	\\
\vdots					& \ddots	& \vdots 						\\
w^{(k)}_{p_1}	& \cdots	& w^{(k)}_{p_{k^\prime}}	\\
\end{matrix}
\right)
\left(
\begin{matrix}
\eps^{\vphantom{(\ell)}}_{p_1} 			\\
\vdots 				\\
\eps^{\vphantom{(\ell)}}_{p_{k^\prime}}	\\
\end{matrix}
\right)
=
\left(
\begin{matrix}
0		\\
\vdots	\\
0		\\
\end{matrix}
\right)
\]
has a non-trivial solution.
But we have already seen that such a situation can never happen, and we can thus
conclude that our method, whose pseudocode we give in
Algorithm~\ref{alg:ssr-multiple-shift}, can also detect up to $k$ errors
occurring in $\vtr{x}$.
Therefore, we have proven the following theorem.

\begin{theorem}[Correctness of Algorithm~\ref{alg:ssr-multiple-shift}]
\label{th:ssr-multiple-shift}
%\begin{comment}
%Let $\vtr{A} \in \R{n \times n}$ be a square matrix, let
%$\vtr{x}, \vtr{y} \in \R{n}$ be the input and output vector respectively, and
%let $\vtr{x}^\prime = \vtr{x}$. Let us assume that the algorithm performs the
%computation
%\begin{equation}
%\widetilde{\vtr{y}} \leftarrow \widetilde{\vtr{A}} \widetilde{\vtr{x}},
%\end{equation}
%where $\widetilde{\vtr{A}} \in \R{n \times n}$ and
%$\widetilde{\vtr{x}} \in \R{n}$ are the possibly faulty representations o
%$\vtr{A}$ and $\vtr{x}$ respectively, while $\widetilde{\vtr{y}} \in \R{n}$ is
%the possibly erroneous computed result of the sparse matrix-vector product.
%Let us also assume that the encoding scheme relies on
%\end{comment}
Let us consider the same notation as in Theorem~\ref{th:ssr-shift}.
Let $\vtr{\underline{W}} \in \R{n + 1 \times n}$ be a matrix such that any
square submatrix is full rank, and let us denote by
\mbox{$\vtr{W} \in \R{n \times n}$} the matrix of its first $n$ rows.
Let us assume an encoding scheme that relies on
\begin{enumerate}
\item an auxiliary checksum matrix
$\vtr{C} = \trans{\left(\trans{\vtr{W}_{\vphantom{0}}}\vtr{A}\right)}$,

\item an auxiliary checksum matrix $\vtr{M} = \vtr{W} - \vtr{C}$,

\item a vector of auxiliary counters $\vtr{s}_{Rowindex}$ initialized to the
null vector and updated at runtime as in lines 16 -- 17 of
Algorithm~\ref{alg:ssr-multiple-shift}),

\item an auxiliary checksum vector
\mbox{$\vtr{c}_{Rowindex} = \trans{\vtr{\underline{W}}}\vtxt{Rowindex}$.}

\end{enumerate}
Then, up to $k$ errors striking the computation of $\vtr{y}$ or the memory
locations that store $\vtr{A}$ or $\vtr{x}$, cause one of the following
conditions to fail:
\begin{enumerate}[i.]
\item $\trans{\vtr{W}} \vtr{y} = \trans{\vtr{C}} \vtr{x}$,
\item $\trans{\vtr{W}}\left( \vtr{x}^\prime - \vtr{y} \right)$,
\item $\vtr{s}_{Rowindex}$ = $\vtr{c}_{Rowindex}$.
\label{it:ssr-mshift-iii}
\end{enumerate}
\end{theorem}

Let us note that we have just shown that our algorithm can detect up to $k$
errors striking only $\vtr{A}$, or only $\vtr{x}$ or only the computation.
Nevertheless, this result holds even when the errors are distributed among the
possible cases, as long as at most $k$ errors rely on the same checkpoint.

It is clear that the execution time of the algorithm depends on both
$\nnz(\vtr{A)}$ and $k$. For the \textsc{computeChecksum} function, the cost is,
assuming that the weight matrix $\vtr{W}$ is already known, $\mO{k \nnz(\vtr{A})}$
for the computation of $\vtr{C}$, and $\mO{k n}$ for the computation of
$\vtr{M}$ and $\vtr{c}_{\vtxt{Rowindex}}$. Hence the number of performed
operations is $\mO{k \nnz(\vtr{A})}$. The overhead added to the SpMxV depends
instead on the computation of four checksum matrices that lead to a number of
operations that grows asymptotically as $k n$.

\begin{algorithm}
\begin{small}
% \begin{spacing}{1.}
\begin{algorithmic}[1]
\Require{$\vtr{A} \in \R {n \times n}$, $\vtr{x} \in \R{n}$}
\Ensure{$\vtr{y} \in \R{n}$ such that $\vtr{y} = \vtr{A} \vtr{x}$ or the detection of up to $k$ errors}

\State{$\vtr{x}^\prime \leftarrow \vtr{x}$}
\State{$[\vtr{\underline{W}}, \vtr{C}, \vtr{M}, \widetilde{\vtr{c}}] =$} \Call{computeChecksums}{$\vtr{A}$, $k$};
\Return{\Call{SpMxV}{$\vtr{A}$, $\vtr{x}$, $\vtr{x}^\prime$, $\vtr{\underline{W}}$, $\vtr{C}$, $\vtr{M}$, $k$, $\widetilde{\vtr{c}}$}};

\Statex

\Function{computeChecksums}{$\vtr{A}$, $k$}
	\State{Generate $\vtr{\underline{W}} = \left[\vtr{w}^{(1)} \dots \vtr{w}^{(k)}\right] \in \R{n+1 \times n}$};
	\State{$\vtr{W} \leftarrow \vtr{\underline{W}}_{1:n, *} \in \R{n \times k}$}
	\State{$\trans{\vtr{C}} \leftarrow \trans{\vtr{W}} \vtr{A}$};
	\State{$\vtr{M} \leftarrow \vtr{W} - \vtr{C}$};
	\State{$\vtr{c}_{\vtxt{Rowindex}} \leftarrow \trans{\vtr{\underline{W}}}\vtxt{Rowindex} $};
	\Return{$\vtr{\underline{W}}, \vtr{C}, \vtr{M}, \vtr{c}_{Rowindex}$};
\EndFunction

\Statex

\Function{SpMxV}{$\vtr{A}$, $\vtr{x}$, $\vtr{x}^\prime$, $\vtr{\underline{W}}$, $\vtr{C}$, $\vtr{M}$, $k$, $\widetilde{\vtr{c}}$}
	\State{$\vtr{W} \leftarrow \vtr{\underline{W}}_{1:n, *} \in \R{n \times k}$}
	\State{$\widetilde{\vtr{s}} \leftarrow [0, \dots, 0]$};
	\For {$i\leftarrow1$ to $n$}
		\State{$y_i \leftarrow 0$};
		\For{$j=1$ to $k$}
			\State{$\widetilde{s}_j \leftarrow \widetilde{s}_j + w_{ij} \vtxt{Rowindex}_i$};
		\EndFor
		\For {$j\leftarrow\vtxt{Rowindex}_i$ to $\vtxt{Rowindex}_{i+1} - 1$}
			\State{$ind \leftarrow \vtxt{Colid}_j$};
			\State{$y_i \leftarrow y_i + \vtxt{Val}_j \cdot x_{ind}$};
		\EndFor
	\EndFor
        \State{$\vtr{d}_{x} \leftarrow \trans{\vtr{W}} \vtr{y} - \trans{\vtr{C}} \vtr{x}$};
        \State{$\vtr{d}_{x^\prime} \leftarrow \trans{\vtr{W}}\left( \vtr{x}^\prime - \vtr{y} \right) - \trans{\vtr{M}} \vtr{x} $};
        \State{$\vtr{d}_r \leftarrow \widetilde{\vtr{c}} - \widetilde{\vtr{s}}$};
	\If{$ \vtr{d}_x = 0 \wedge \vtr{d}_{x^\prime} = 0 \wedge \vtr{d}_r = 0$}
		\Return{$\vtr{y}$};
	\Else
		\State{\textbf{error} (``Soft errors are detected")};
	\EndIf
\EndFunction
\end{algorithmic}
\caption{Shifting checksum algorithm for $k$ errors detection.}
\label{alg:ssr-multiple-shift}
\end{small}
\end{algorithm}

\subsection{Single error correction\label{sec.err.correct}}

We now discuss single error correction, using
Algorithm~\ref{alg:ssr-multiple-shift} as a reference.
We describe how a single error striking either memory or computation can be
not only detected but also corrected at line~27.
We use only two checksum vectors, that is, we describe correction of single
errors assuming that two errors cannot strike the same SpMxV.
By the end of the section, we will generalize this approach and discuss how
single error correction and double error detection can be performed concurrently
by exploiting three linearly independent checksum vectors.
Whenever a single error is detected, regardless of its location (computation or
memory), it is corrected by means of a succession of various steps.
When one or more errors are detected, the correction mechanism tries to
determine their multiplicity and, in case of a single error, what memory
locations have been corrupted or what computation has been miscarried.
Errors are then corrected using the values of the checksums and, if
need be, partial recomputations of the result are performed.

%%Up to this point, we have considered detection and correction as separate
%%problems, but we need a way to combine the two approaches in order to implement
%%\abftcorrect as defined in Section~\ref{sec.CG}.
%%Indeed, note that double errors could be shadowed when using
%%Algorithm~\ref{alg:ssr-shift}, although the probability of such an event is
%%negligible.

% We will use the matrix

% \[
% \vtr{W} =
% \begin{bmatrix}
% 1 & 1 \\
% 1 & 2 \\
% 1 & 3 \\
% \vdots & \vdots \\
% 1 & n
% \end{bmatrix},
% \]
% as weight matrix in Algorithm~\ref{alg:ssr-multiple-shift}. This matrix
% is enough to guarantee, on the one hand, single error correction,
% and, on the other hand, two errors detection, but not to achieve both correction and
% detection at the same time.
% This way we achieve both protections with an overhead of just three vector
% operations per SpMxV.

% For correction we proceed as follows.
As we did for multiple error detection, we require that any $2 \times 2$
submatrix of $\vtr{W} \in \R{n \times 2}$ has full rank.
The simplest example of weight matrix having this property is probably
\[
\vtr{W} =
\begin{bmatrix}
1 & 1 \\
1 & 2 \\
1 & 3 \\
\vdots & \vdots \\
1 & n
\end{bmatrix}.
\]

To detect errors striking $\vtxt{Rowidx}$, we compute the ratio $\rho$ of
the second component of $\vtr{d}_r$ to the first one, and check whether its
distance from an integer is smaller than a certain threshold parameter $\eps$.
If this distance is smaller, the algorithm concludes that the $\sigma$th element, where $\sigma=Round(\rho)$ is the nearest integer to $\rho$, of $\vtxt{Rowidx}$
is faulty, performs the correction by subtracting the first component of
$\vtr{d}_r$ from $\vtxt{Rowidx}_\sigma$, and recomputes $y_{\sigma}$ and $y_{\sigma - 1}$, if the
error in $\vtxt{Rowindex_\sigma}$ is a decrement; or $y_{\sigma + 1}$ if it was an
increment.
%%%%\bucomment{how do we know if the error is an increment or decrement?}.
Otherwise, it just emits an error.

The correction of errors striking $\vtxt{Val}$, $\vtxt{Colid}$ and the
computation of $y$ are performed together.
Let now $\rho$ be the ratio of the second component of $\vtr{d}_x$ to the first
one.
If $\rho$ is near enough to an integer $\sigma$, the algorithm computes the checksum matrix
$\vtr{C}^\prime = \trans{\vtr{W}} \vtr{A}$ and considers the number
$z_{\widetilde{\vtr{C}}}$ of non-zero columns of the difference matrix
$\widetilde{\vtr{C}} = | \vtr{C} - \vtr{C}^\prime |$.
At this stage, three cases are possible:
\begin{itemize}
\item If $z_{\widetilde{\vtr{C}}} = 0$, then the error is in the computation of
$y_\sigma$, and can be corrected by simply recomputing this value.

\item If $z_{\widetilde{\vtr{C}}} = 1$, then the error has struck an element of
$\vtxt{Val}$.
Let us call $f$ the index of the non-zero column of $\widetilde{\vtr{C}}$.
The algorithm finds the element of $\vtxt{Val}$ corresponding to the entry at
row $\sigma$ and column $f$ of $A$ and corrects it by using the column checksums
much like as described for $\vtxt{Rowidx}$.
Afterwards, $y_d$ is recomputed to fix the result.

\item If $z_{\widetilde{\vtr{C}}} = 2$, then the error concerns an element of
$\vtxt{Colid}$. Let us call $f_1$ and $f_2$ the index of the two non-zero
columns and $m_1$, $m_2$ the first and last elements of $\vtxt{Colid}$
corresponding to non-zeros in row $\sigma$. It is clear that there exists exactly one
index $m^*$ between $m_1$ and $m_2$ such that either $\vtxt{Colid}_{m^*} = f_1$
or $\vtxt{Colid}_{m^*} = f_2$. To correct the error it suffices to switch the
current value of $\vtxt{Colid}_{m^*}$, i.e., putting $\vtxt{Colid}_{m^*} = f_2$
in the former case and $\vtxt{Colid}_{m^*} = f_1$ in the latter. Again, $y_\sigma$
has to be recomputed.

\item if $z_{\widetilde{\vtr{C}}} > 2$, then errors can be detected but not
corrected, and an error is emitted.
\end{itemize}

To correct errors striking $\vtr{x}$, the algorithm computes $\rho$, that is the
ratio of the second component of $\vtr{d}_{x^{\prime}}$ to the first one, and
checks that the distance between $d$ and the nearest integer $\sigma$ is smaller than
$\eps$.
Provided that this condition is verified, the algorithm computes the value of
the error $\tau = \sum_{i = 1}^{n}x_i - cx_\sigma$ and corrects $x_\sigma = x_\sigma - \tau$.
The result is updated by subtracting from $\vtr{y}$ the vector
$\vtr{y}^\tau = \vtr{A} \vtr{x}^\tau$, where $\vtr{x}^\tau \in \R{n \times n}$
is such that $x_\sigma^\tau = \tau$ and $x_i^\tau = 0$ otherwise.

Let us now investigate how detection and correction can be combined and let us
give some details about the implementation of \abftcorrect as defined in
Section~\ref{sec.CG}.
Indeed, note that double errors could be shadowed when using
Algorithm~\ref{alg:ssr-shift}, although the probability of such an event is
negligible.

Let us restrict ourselves to an easy case, without considering errors in
$\vtr{x}$.
As usual, we compute the column checksums matrix
\[
\vtr{C} = \trans{\left(\trans{\vtr{W}}\vtr{A}\right)},
\]
and then compare the two entries of $\trans{\vtr{C}}\vtr{x} \in \R{2}$
with the weighted sums
\[
\widetilde{y}^c_1 = \sum_{i=1}^{n} \widetilde{y}_i
\]
and
\[
\widetilde{y}^c_2 = \sum_{i=1}^{n} i\widetilde{y}_i
\]
where $\widetilde{\vtr{y}}$ is the possibly faulty vector computed by the
algorithm.
It is clear that if no error occurs, the computation verifies the condition
$\vtr{\delta} = \widetilde{\vtr{y}}^c - \vtr{c} = 0$.
Furthermore, if exactly one error occurs, we have
$\delta_1, \delta_2 \neq 0$ and $\frac{\delta_2}{\delta_1} \in \N{}$, and if two
errors strike the vectors protected by the checksum $\vtr{c}$, the algorithm is
able to detect them by verifying that $\vtr{\delta} \neq 0$.

At this point it is natural to ask whether this information is enough to build
a working algorithm or some border cases can bias its behavior.
In particular, when $\frac{\delta_2}{\delta_1} = p \in \N{}$, it is not clear how
to discern between single and double errors.
Let $\eps_1, \eps_2 \in \R{} \setminus \left\lbrace 0 \right\rbrace$ be the
value of two errors occurring at position $p_1$ and $p_2$ respectively, and let
$\widetilde{\vtr{y}} \in \R{n}$ be such that
\[
\widetilde{y}_i =
\left\lbrace
\begin{array}{ll}
y_i, \qquad\qquad &1 \leq i \leq n, i \neq p_1, p_2 \\
y_i + \eps_1, &i = p_1 \\
y_i + \eps_2, &i = p_2
\end{array}
\right.
.
\]
Then the conditions
\begin{eqnarray}
\delta_1 &=& \eps_1 + \eps_2, \\
\delta_2 &=& p_1 \eps_1 + p_2 \eps_2,
\end{eqnarray}
hold. Therefore, if $\eps_1$ and $\eps_2$ are such that
\begin{equation}
\label{eq:line-condition}
p \left(\eps_1 + \eps_2\right) = p_1 \eps_1 + p_2 \eps_2,
\end{equation}
it is not possible to distinguish these two errors from a single error of value
$\eps_1 + \eps_2$ occurring in position $p$. This issue can be solved by
introducing a new set of weights and hence a new row of column checksums.
Let us consider a weight matrix $\widehat{\vtr{W}} \in \R {n \times 3}$
that includes a quadratic weight vector
\[
\widehat{\vtr{W}} =
\begin{bmatrix}
1 & 1 & 1 \\
1 & 2 & 4 \\
1 & 3 & 9 \\
\vdots & \vdots  & \vdots\\
1 & n & n^2
\end{bmatrix},
\]
and the tridimensional vector
\[
\hat{\delta} = \left(\trans{\widehat{\vtr{W}}}\vtr{A}\right)\vtr{x} -
	\trans{\widehat{\vtr{W}}}\widetilde{\vtr{y}}\;,
\]
whose components can be expressed as
\begin{eqnarray}
\delta_1 &=& \eps_1 + \eps_2 \nonumber \\
\delta_2 &=& p_1 \eps_1 + p_2 \eps_2 \nonumber \\
\delta_2 &=& p_1^2 \eps_1 + p_2^2 \eps_2\;. \nonumber
\end{eqnarray}
To be confused with a single error in position $p$, $\eps_1$ and $\eps_2$ have
to be such that
\[
p \left(\eps_1 + \eps_2\right) = p_1 \eps_1 + p_2 \eps_2
\]
and
\[
p^2 \left(\eps_1 + \eps_2\right) = p_1^2 \eps_1 + p_2^2 \eps_2
\]
hold simultaneously for some $p \in \N{}$. In other words, possible values of
the errors are the solution of the linear system
\[
\begin{pmatrix}
(p - p_1) & (p - p_2) \\
(p^2 -p_1^2) & (p^2 - p_2^2) \\
\end{pmatrix}
\begin{pmatrix}
\eps_1 \\
\eps_2 \\
\end{pmatrix}
=
\begin{pmatrix}
0 \\
0
\end{pmatrix}.
\]
It is easy to see that the determinant of the coefficient matrix is
\[
\left(p - p_1\right)\left(p - p_2\right)\left(p_2 - p_1\right),
\]
which always differs from zero, as long as $p$, $p_1$ and $p_s$ differ pairwise.
Thus, the matrix is invertible, and the solution space of the linear system is
the trivial kernel $(\eps_1, \eps_2) = (0, 0)$. Thus using $\widehat{\vtr{W}}$
as weight matrix guarantees that it is always possible to distinguish a single error
from double errors.

%\begin{comment}
%Let us analyse further the conditions that are necessary for such an issue to
%arise.
%From the equality in Equation~\eqref{eq:line-condition} follows that two errors
%in position $p_1$ and $p_2$ can be confused with an error in position $p$ if and
%only if their respective values $\eps_1$ and $\eps_2$ are such that
%\begin{equation}
%\label{eq:line}
%\eps_1 = \dfrac{\left( p - p_1 \right)}{\left( p_1 - p \right)}\eps_2.
%\end{equation}
%Equation~\eqref{alg:ssr} is the equation of a line in $\R{2}$ that intercepts 0.
%\end{comment}

\section{Performance model}
\label{sec.maths}
In Section~\ref{sec.model.general}, we introduce the general performance model.
Then in Section~\ref{sec.model.CG} we instantiate it for the three methods that
we are considering, namely
\onlinechen, \abftdetect and \abftcorrect.

\subsection{General approach}
\label{sec.model.general}
We introduce an abstract performance model to compute the best
combination of checkpoints and verifications for iterative methods.
We execute $T$ time-units of work followed by a verification,
which we call a \emph{chunk},
and we repeat this scheme
$s$ times, i.e., we compute $s$ chunks, before taking a checkpoint.
We say that the $s$ chunks constitute a \emph{frame}.
The whole execution is then partitioned into frames.
We assume that checkpoint, recovery and verification are error-free
operations. Let $T_{cp}$, $T_{rec}$ and $T_{verif}$ be the respective cost
of these operations. Finally, assume an exponential distribution  of errors
and let $q$ be the probability of successful execution
for each chunk: $q = e^{-\lambda T}$, where $\lambda$ is the fault rate.

The goal of this section is to compute the expected time $\ep{s,T}$
needed to execute a frame composed of $s$ chunks of size $T$.
We derive the best value of $s$ as a function of $T$ and of the resilience
parameters $T_{cp}$, $T_{rec}$, $T_{verif}$, and $q$, the success probability of
a chunk.
Each frame is preceded by a checkpoint, except maybe the first one (for which we
recover by reading initial data again).
Following earlier work~\cite{c178}, we derive the following recursive equation
to compute the expected completion time of a single frame:
\begin{align}
  \label{eq:newmodel}
  \begin{split}
    \ep{s,T} &= q^{s} ( s (T+T_{verif})) + T_{cp}) \\
    &+ \left( 1 - q^{s} \right) \left(\ep{T_{lost}} + T_{rec} + \ep{s,T}\right).
  \end{split}
\end{align}

Indeed, the execution is successful if all chunks are successful, which happens
with probability $q^{s}$, and in this case the execution time simply is the sum
of the  execution times of each chunk plus the final checkpoint.
Otherwise, with probability $1-q^{s}$, we have an error, which we detect after
some time $\ep{T_{lost}}$, and that forces us to recover (in time $T_{rec}$) and restart the frame anew,
hence in time $\ep{s,T}$.
The difficult part is to compute $\ep{T_{lost}}$.

For $1 \leq i \leq s$,
let $\fff_{i}$ be the following conditional probability:
%\begin{equation*}
\begin{multline}
\fff_{i} = \mathbb{P}(\text{error strikes at chunk $i$}|\text{there is }\\
\text {an error in the frame})\;.
\end{multline}
%\end{equation*}
Given the success probability $q$ of a chunk, we obtain that
$$\fff_{i} = \frac{q^{i-1} (1-q)}{1-q^{s}}\;.$$
Indeed, the first $i-1$ chunks were successful (probability $q^{i-1}$), the
$i$th one had an error (probability $1-q$), and we condition by the probability
of an error within the frame, namely $1-q^{s}$. With probability $\fff_{i}$, we
detect the error at the end of the $i$th chunk, and we have lost the time spent
executing the first $i$ chunks.
We derive that
$$ \ep{T_{lost}} = \sum_{i=1}^{s} \fff_{i} \left( i (T+T_{verif}) \right)\;.$$

We have
$\sum_{i=1}^{s} \fff_{i} = \frac{(1-q)h(q)}{1-q^{s}}$ where
$h(q) = 1+2q+3q^{2}+\cdots+ s q^{s-1}$.
If $m(q) = q+q^{2}+\cdots +q^{s} = \frac{1-q^{s+1}}{1-q} - 1$,
we get by differentiation that $m'(q)=h(q)$, hence
$h(q)=\frac{- (s+1) q^{s}}{1-q} + \frac{1-q^{s+1}}{(1-q)^{2}}$
and finally
$$\ep{T_{lost}} = (T+T_{verif}) \frac{s q^{s+1} -(s+1) q^{s} +1}{(1-q^{s})(1-q)}\;.$$
Plugging the expression of  $ \ep{T_{lost}}$ back into~\eqref{eq:newmodel},
we obtain
%\begin{footnotesize}
\begin{align*}
\ep{s,T} &= s (T+T_{verif})) + T_{cp} + (q^{-s} -1) T_{rec} \nonumber \\
&+ T \frac{s q^{s+1} -(s+1) q^{s} +1}{q^{s}(1-q)}\;, \nonumber
\end{align*}
%\end{footnotesize}
\noindent
which simplifies into
\begin{equation*}
\ep{s,T} = T_{cp} + (q^{-s} -1) T_{rec} +  (T+T_{verif}) \frac{1-q^{s}}{q^{s}(1-q)}\;.
\end{equation*}
We have to determine the value of $s$ that minimizes the overhead of a frame:
\begin{equation}
\label{eq:minimize}
s = \argmin_{s \geq 1}\left(
\dfrac{\ep{s,T}}{s T} \right).
\end{equation}
The minimization is complicated and should be conducted numerically (because
$T$, the size of a chunk, is still unknown).
Luckily, a dynamic programming algorithm to compute the optimal value of $T$
and $s$ is available~\cite{pmbs2014}.

\subsection{Instantiation to PCG}
\label{sec.model.CG}
For each of the three methods, \onlinechen, \abftdetect and \abftcorrect,
we instantiate the previous model and discuss how to solve~\eqref{eq:minimize}.

\subsubsection{\onlinechen}

For Chen's method~\cite{Chen2013}, we have chunks of $d$ iterations,
hence $T = d T_{iter}$, where $T_{iter}$ is the raw cost of a PCG iteration
without any resilience method.
The verification time $T_{verif}$ is the cost of the orthogonality check
operations performed as described in Section~\ref{sec.CG}.
As for silent errors, the application is protected from arithmetic errors
in the ALU, as in Chen's original method, but also for corruption in data memory
(because we also checkpoint the matrix $\vtr{A}$).
Let $\lambda_{a}$ be the rate of arithmetic errors, and   $\lambda_{m}$ be the
rate of memory errors.
For the latter, we have $\lambda_{m} = M \lambda_{word}$ if the data memory
consists of $M$ words, each susceptible to be corrupted with rate
$\lambda_{word}$.
Altogether, since the two error sources are independent, they have a cumulative
rate of $\lambda = \lambda_{a}+\lambda_{m}$, and the success probability for a
chunk is $q = e^{-\lambda T}$.

Plugging these values in~\eqref{eq:minimize} gives an optimization
formula very similar to that of Chen~\cite[Sec.~5.2]{Chen2013}, the only
difference being that we assume that the verification is error-free, which is
needed for the correctness of the approach.

\subsubsection{\abftdetect}

When using ABFT techniques, we detect possible errors every iteration,
so a chunk is a single iteration, and $T= T_{iter}$.
For \abftdetect, $T_{verif}$
is the overhead due to the checksums and redundant operations to detect a single error 
in the method.

\abftdetect can protect the application
from the same silent errors as \onlinechen, and just as before
the success probability for a chunk (a single iteration here) is
$q = e^{-\lambda T}$.

\subsubsection{\abftcorrect}

In addition to detection, we now correct single errors at every chunk.
Just as for \abftdetect, a chunk is a single iteration, and $T= T_{iter}$,
but $T_{verif}$ corresponds to a larger overhead, mainly due to the
extra checksums needed to detect two errors and correct a single one.

The main difference lies in the error rate. An iteration with \abftcorrect
is successful if zero or one error has struck during that iteration, so that the
success probability is much higher than for \onlinechen and \abftdetect.
We compute that value of the success probability as follows.
We have a Poisson process of rate $\lambda$, where
$\lambda = \lambda_{a}+\lambda_{m}$ as for
\onlinechen and \abftdetect.
The probability of exactly $k$ errors in time $T$ is $\frac{(\lambda T)^{k}}{k!} e^{-\lambda T}$~\cite{Mitzenmacher2005}, hence the probability of no error is
$e^{-\lambda T}$ and the probability of exactly one error is $\lambda T e^{-\lambda T}$,
so that $q =e^{-\lambda T} + \lambda T e^{-\lambda T}$.

\section{Experiments}
\label{sec.experiments}
\subsection{Setup}
There are two different sources of advantages in combining ABFT and
checkpointing.
First, the error detection capability lets us perform a cheap validation of the
partial result of each PCG step, recovering as soon as an error strikes.
Second, being able to correct single errors makes each step more resilient and
increases the expected number of consecutive valid iterations.
We say an iteration is valid if it is non-faulty, or if it suffers from a single
error that is corrected via ABFT.

For our experiments, we use a set of positive definite matrices from the UFL
Sparse Matrix Collection~\cite{florida}, with size between 17456 and 74752 and
density lower than $10^{-2}$.
We perform the experiments under Matlab and use the factored approximate inverse
preconditioners~\cite{Benzi1998,bzkt:01} in the PCG.
The application of these preconditioners requires two SpMxV, which are protected
against error using the methods proposed in Section~\ref{s:abftspmxv} (in all
methods \onlinechen, \abftdetect, and \abftcorrect).

At each iteration of PCG, faults are injected during vector and matrix-vector
operations but, since we are assuming selective reliability, all the checksums
and checksum operations are considered non-faulty.
Faults are modeled as bit flips occurring independently at each step, under an
exponential distribution of parameter $\lambda$, as detailed in
Section~\ref{sec.model.CG}.
These bit flips can strike either the matrix (the elements of $\vtxt{Val},
\vtxt{Colid}$ and $\vtxt{Rowidx}$), or any entry of the PCG vectors
$\vtr{r}_i$, $\vtr{z}_i$, $\vtr{q}$, $\vtr{p}_i$ or $\vtr{x}_i$.
We chose not to inject errors during the computation explicitly, as they are
just a special case of error we are considering.
Moreover, to simplify the injection mechanism, $T_{iter}$ is normalized to be
one, meaning that each memory location or operation is given the chance to fail
just once per iteration~\cite{Sao2013}.
Finally, to get data that are homogeneous among the test matrices, the fault
rate $\lambda$ is chosen to be inversely proportional to $M$ (memory size) with
a proportionality constant \mbox{$\alpha \in \left( 0, 1 \right)$}; this makes
sense as larger the memory used by an application, larger is the chance to have
an error.
It follows that the expected number of PCG iterations between two distinct fault
occurrences does not depend either on the size or on the sparsity ratio of the
matrix.

We compare the performance of three algorithms, namely \onlinechen, \abftdetect
(single detection and rolling back as soon as an error is detected), and
\abftcorrect (correcting single errors during a given step and rolling back only
if two errors strike a single operation).
We instantiate them by limiting the maximum number of PCG steps to 50
(20 for \#924, whose convergence is sublinear) and setting the tolerance
parameter $\epsilon$ at line 5 of Algorithm~\ref{alg:cg} to $10^{-14}$.
The number of iterations for a non-faulty execution and the achieved accuracy
are detailed in Table~\ref{tab:matrices}.

\begin{table}
  \centering
  \caption{Test matrices used in the experiments. Name and id are from the University of Florida Sparse Matrix Collection.}
  \begin{tabular}[t]{lrrrrr}
    \toprule
    name & id & size & density & steps & residual \\
    \midrule
    Boeing/bcsstk36 & 341 & 23052 & 2.15e-03 & 50 & 6.41e-04 \\
    Mulvey/finan512 & 752 & 74752 & 1.07e-04 & 25 & 2.19e-14 \\
    Andrews/Andrews & 924 & 60000 & 2.11e-04 & 20 & 1.59e-04 \\
    GHS\_psdef/wathen100 & 1288 & 30401 & 5.10e-04 & 50 & 2.55e-13 \\
    GHS\_psdef/wathen120 & 1289 & 36441 & 4.26e-04 & 50 & 9.16e-14 \\
    GHS\_psdef/gridgena & 1311 & 48962 & 2.14e-04 & 50 & 5.61e-05 \\
    GHS\_psdef/jnlbrng1 & 1312 & 40000 & 1.24e-04 & 50 & 5.83e-13 \\
    UTEP/Dubcova2 & 1848 & 65025 & 2.44e-04 & 50 & 1.16e-05 \\
    JGD\_Trefethen/Trefethen\_20000 & 2213 & 20000 & 1.39e-03 & 10 & 6.00e-16 \\
    \bottomrule
  \end{tabular}
  \label{tab:matrices}
\end{table}

Implementing the null checks in Algorithm~\ref{alg:ssr-shift},
Algorithm~\ref{alg:ssr-split} and Algorithm~\ref{alg:ssr-multiple-shift}
poses a challenge.
The comparison $\vtr{d}_r = 0$ is between two integers, and can be correctly
evaluated by any programming language using the equality check.
However, the other two, having floating point operands, are problematic.
Since the floating point operations are not associative and the distributive
property does not hold, we need a tolerance parameter that takes into account
the rounding operations that are performed by each floating point operation.
Here, we give an upper bound on the difference between the two floating point
checksums, using the standard model~\cite[Sec.~2.2]{higham} to make sure
that errors caught by our algorithms really are errors and not merely
inaccuracies due to floating point operations (which is tolerable, as non-faulty
executions can give rise to the same inaccuracy).

\begin{theorem}[Accuracy of the floating point weighted checksums]
\label{th:gen-fp-bounds}
Let $\vtr{A} \in \R{n \times n}$, $\vtr{x} \in \R{n}$, $\vtr{c} \in \R{n}$.
If all of the sums involved into the matrix operations are performed using some
flavor of recursive summation~\cite[Ch.~4]{higham}, it holds that
\begin{equation}
\label{eq:gen-fp-bound-ssr}
| \fl\left(\left(\trans{\vtr{c}}\vtr{A}\right)\vtr{x}\right) - \fl\left(\trans{\vtr{c}}\left(\vtr{A}\vtr{x}\right)\right) |
\leq
2 \; \gamma_{2n}\; | \trans{\vtr{c}}|\;\mid\vtr{A} \mid\;\mid\vtr{x}\mid\;.
\end{equation}
\label{th:fp2}
\end{theorem}

We refer the reader to the technical report for the proof~\cite[Theorem 2]{ens-report}.
Let us note that if all of the entries of $\vtr{c}$ are positive, as it is often
the case in our setting, the absolute value of $\vtr{c}$
in~\eqref{eq:gen-fp-bound-ssr} can be safely replaced with $\vtr{c}$ itself.
It is also clear that these bounds are not computable, since
$\trans{\vtr{c}} \mid \vtr{A} \mid \; \mid \vtr{x} \mid$ is not, in general, a
floating point number.
This problem can be alleviated by overestimating the bound by means of matrix
and vector norms.

Since we are interested in actually computing the bound at runtime, we
consider a weaker bound.
Recalling that~\cite[Sec. B.7]{higham-fom}
\begin{equation}
\label{eq:matrix-inf-norm}
\norm{\vtr{A}}_1 = \max_{1 \leq j \leq n} \sum_{i=1}^{n} | a_{i,j} |.
\end{equation}
we can upper bound the right hand side in so that
\begin{equation}
\label{eq:fp-computable-bound}
|\fl\left(\left(\trans{\vtr{c}}\vtr{A}\right)\vtr{x}\right) - \fl\left(\trans{\vtr{c}}\left(\vtr{A}\vtr{x}\right)\right) |
\leq
2\;\gamma_{2n} \; n \; \norm{\trans{\vtr{c}}}_\infty \norm{\vtr{A}}_1 \norm{\vtr{x}}_\infty\;.
\end{equation}
Though the right hand side of~\eqref{eq:fp-computable-bound} is not exactly
computable in floating point arithmetic, it requires an amount of operations
dramatically smaller than~\eqref{eq:gen-fp-bound-ssr}; just a few sums for the
norm of $\vtr{A}$.
As this norm is usually computed using the identity
in~\eqref{eq:matrix-inf-norm}, any kind of summation yields a relative error of
at most $n^\prime\vtr{u}$~\cite[Sec.~4.6]{higham}, where $n^\prime$ is the
maximum number of nonzeros in a column of $\vtr{A}$, and $\vtr{u}$ is the
machine epsilon.
Since we are dealing with sparse matrices, we expect $n^\prime$ to be very
small, and hence the computation of the norm to be accurate.
Moreover, since the right hand side in~\eqref{eq:fp-computable-bound} does not
depend on $\vtr{x}$, it can be computed just once for a given matrix and weight
vector.

Clearly, using~\eqref{eq:fp-computable-bound} as tolerance parameter guarantees
no false positive (a computation without any error that is considered as
faulty), but allows false negatives (an iteration during which an error occurs
without being detected) when the perturbations of the result are small.
Nonetheless, this solution works almost perfectly in practice, meaning that
though the convergence rate can be slowed down, the algorithms still converges
towards the ``correct'' answer.
Though such an outcome could be surprising at first, Elliott et
al.~\cite{fp-quant, fp-quant-rep} showed that bit flips that strike the less
significant digits of the floating point representation of vector elements
during a dot product create small perturbations of the results, and that the
magnitude of this perturbation gets smaller as the size of the vectors
increases.
Hence, we expect errors that are not detected by our tolerance threshold to be
too small to impact the solution of the linear solver.

%\newcolumntype{h}{>{\bfseries}r}
\begin{table}
\caption{Experimental validation of the model. Here $\widetilde{s_i}$ and $s^*_i$ represent the best checkpointing interval according to our model and to our simulations respectively, whereas $\et{\widetilde{s_i}}$ and $\et{s^*_i}$ stand for the execution time of the algorithm using these checkpointing intervals.}
%\footnotesize{
%\begin{tabular*}{\textwidth}{@{\extracolsep{\fill}}rrr|rrrrh|rrrrh}
\begin{tabular}{r|rrrrr|rrrrr}
\toprule
id & $\widetilde{s_1}$ & $\et{\widetilde{s_1}}$ & $s_1^*$ & $\et{s_1^*}$ & $l_1$ & $\widetilde{s_2}$ & $\et{\widetilde{s_2}}$ & $s_2^*$ & $\et{s_2^*}$ & $l_2$ \\ 
\midrule
341 & 4 & 305.42 & 1 & 293.22 & \bf{4.16} & 4 & 305.45 & 1 & 293.16 & \bf{4.19} \\
752 & 30 & 13.81 & 24 & 13.34 & \bf{3.57} & 24 & 12.17 & 23 & 11.93 & \bf{2.01} \\
924 & 23 & 49.82 & 30 & 47.53 & \bf{4.82} &  23 & 44.52 & 26 & 42.42 & \bf{4.96} \\
1288 & 22 & 11.12 & 19 & 10.82 & \bf{2.72} & 22 & 11.32 & 19 & 11.03 & \bf{2.58} \\
1289 & 16 & 16.56 & 13 & 16.38 & \bf{1.07} & 23 & 13.49 & 23 & 13.49 & \bf{0.00} \\
1311 & 4 & 216.70 & 1 & 207.97 & \bf{4.19} & 4 & 220.20 & 1 & 208.19 & \bf{5.77} \\
1312 & 25 & 14.41 & 22 & 13.86 & \bf{3.97} & 23 & 12.30 & 22 & 12.06 & \bf{1.96 }\\
1848 & 4 & 321.70 & 1 & 309.28 & \bf{4.01} & 4 & 366.03 & 1 & 314.20 & \bf{16.49} \\
2213 & 19 & 2.31 & 12 & 2.19 & \bf{5.58} & 24 & 2.33 & 23 & 2.20 & \bf{5.94} \\
\bottomrule
\end{tabular}
\label{tab:exp12}
\end{table}

\subsection{Simulations}
To validate the model, we perform the simulation whose results are illustrated
in Table~\ref{tab:exp12}.
For each matrix, we set $\lambda = \frac{1}{16\;M}$ and consider the average
execution time of 100 repetitions of both \abftdetect (columns 5-8) and
\abftcorrect (columns 6-9). In the table we record the checkpointing interval
$s^*_i$ which achieves the shortest execution time $\et{s_1^*}$, and the
checkpointing interval $\widetilde{s_i}$ which is the best stepsize according to
our method, along with its execution time $\et{\widetilde{s_i}}$.
Finally, we evaluate the performance of our guess by means of the quantity
\[
l_i = \dfrac{\et{\widetilde{s_i}} - \et{s^*_i}}{\et{s^*_i}} \cdot 100\;,
\]
that expresses the loss, in terms of execution time, of executing with the
checkpointing interval given by our model with respect to the best possible
choice.

From the table, we clearly see that the values of $\widetilde{s_i}$ and $s^*_i$ are
close, since the time loss reaches just above 5\% for $l_1$ and just below 15\% for $l_2$.
This sometimes poor result depends just on the small number of repetitions we
are considering, that leads to the presence of outliers, lucky runs in which a
small number of errors occur and the computation is carried on in a much quicker
way.
Similar results hold for other values of $\lambda$.

\pgfplotsset{
	every axis/.append style={font=\tiny},
}
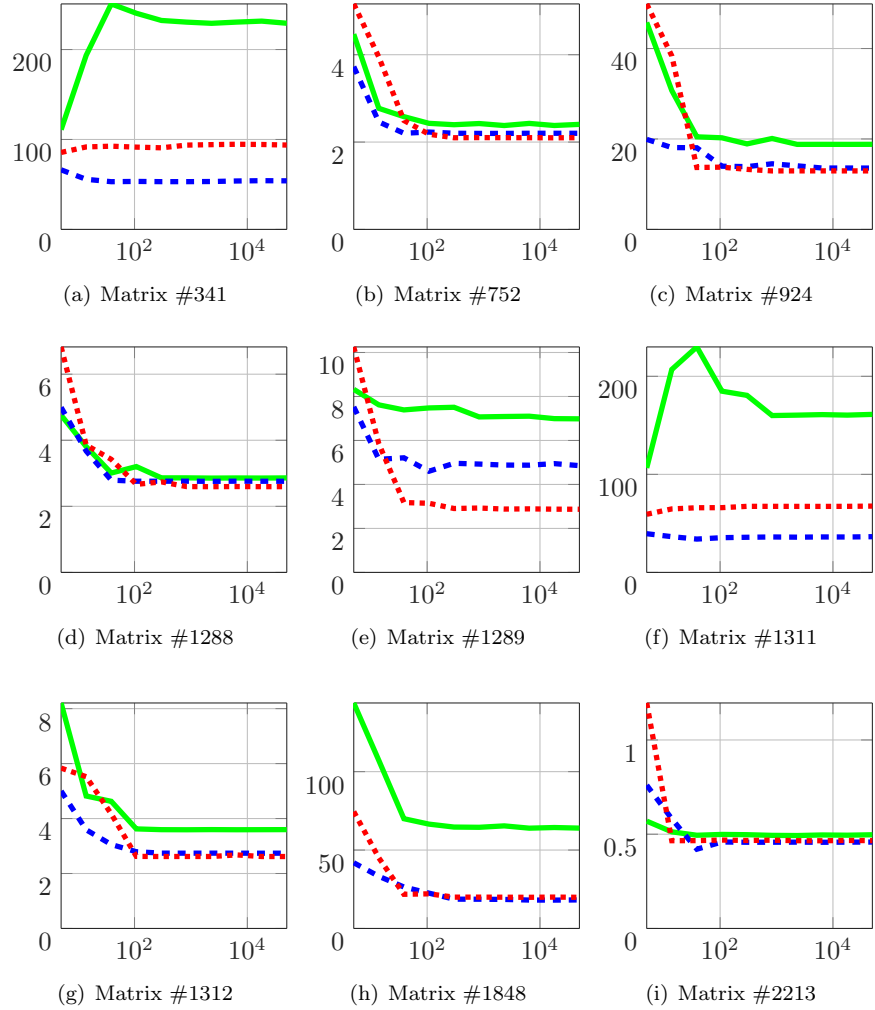
\begin{figure}
\centering
\subfigure[Matrix \#341]{\input{figures/comparison_341.tikz}}
\subfigure[Matrix \#752]{\input{figures/comparison_752.tikz}}
\subfigure[Matrix \#924]{\input{figures/comparison_924.tikz}}

\subfigure[Matrix \#1288]{\input{figures/comparison_1288.tikz}}
\subfigure[Matrix \#1289]{\input{figures/comparison_1289.tikz}}
\subfigure[Matrix \#1311]{\input{figures/comparison_1311.tikz}}

\subfigure[Matrix \#1312]{\input{figures/comparison_1312.tikz}}
\subfigure[Matrix \#1848]{\input{figures/comparison_1848.tikz}}
\subfigure[Matrix \#2213]{\input{figures/comparison_2213.tikz}}
\caption{Execution time in seconds ($y$ axis) of \onlinechen (\textcolor{red}{dotted}), \abftdetect (\textcolor{green}{solid line}) and \abftcorrect (\textcolor{blue}{dashed})  with respect to the normalized MTBF ($x$-axis). The matrix number is in the subcaption.}
\label{fig:exp3}
\end{figure}

We also compare the execution time of the three algorithms to empirically asses how much their relative performance depend on the fault rate. The results on our test matrices are shown in Fig.~\ref{fig:exp3}, where the y-axis is the execution time (in seconds), and the x-axis is the normalized mean time between failure (the reciprocal of $\alpha$). Here, the larger $x = \frac{1}{\alpha}$, the smaller the corresponding value of $\lambda = \frac{\alpha}{M}$, hence the smaller the expected number of errors.
For each value of $\lambda$, we draw the average execution time of 50 runs of the three algorithms, using the best checkpointing interval predicted in 
Section~\ref{sec.model.general} for \abftdetect and \abftcorrect,
and by \mbox{Chen~\cite[Eq. 10]{Chen2013}} for \onlinechen. In terms of execution time, Chen's method is comparable to ours for middle to high fault rates, since it clearly outperforms \abftdetect in five out of nine cases, being slightly faster than \abftcorrect for lower fault rates.
%%%%\todo{The behavior of the methods for \#341, \#1311 and \#1848 depends -- I guess
%%%%-- on the fact that just a bunch of very slow iterations are required to
%%%%converge.
%%%%Note that we achieve best speed by checkpointing at after each iteration.
%%%%In this case, our method consistently suggest to go for a checkpoint every
%%%%fourth iteration.
%%%%I think this is good, but the result \textit{per se} is not so interesting,
%%%%perhaps.
%%%%We might want to keep those matrices or to drop them, I'm fine with either
%%%%choice.
%%%%}

Intuitively, this behavior is not surprising. When $\lambda$ is large, many errors occur but, since $\alpha$ is between zero and one, we always have, in expectation, less than one error per iteration. Thus \abftcorrect requires fewer checkpoints than \abftdetect and almost no rollback, and this compensates for the slightly longer execution time of a single step. When the fault rate is very low, instead, the algorithms perform almost the same number of iterations, but \abftcorrect takes slightly longer due to the additional \mbox{dot products} at each step.

Altogether, the results show that \abftcorrect outperforms both \onlinechen and \abftdetect for a wide range of fault rates, thereby demonstrating that combining checkpointing with ABFT correcting techniques is more efficient than pure checkpointing for most practical situations.

\section{Conclusion}
\label{sec.conclusion}
We consider the problem of silent errors in iterative linear systems solvers.
At first, we focus our attention on ABFT methods for SpMxV, developing
algorithms able to detect and correct errors in both memory and computation
using various checksumming techniques. Then, we combine ABFT with replication,
in order to develop a resilient PCG kernel that can protect axpy's and dot
products as well.
We also discuss how to take numerical issues into account when dealing with
actual implementations.
These methods are a worthy choice for a selective reliability model, since most
of the operations can be performed in unreliable mode, whereas only checksum
computations need to be performed reliably.

In addition, we examine checkpointing techniques as a tool to improve the
resilience of our ABFT PCG and develop a model to trade-off the checkpointing
interval so to achieve the shortest execution time in expectation.
We implement two of the possible combinations, namely an algorithm that relies
on roll back as soon as an error is detected, and one that is able to correct a
single error and recovers from a checkpoint just when two errors strike.
We validate the model by means of simulations and finally compare our algorithms
with Chen's approach, empirically showing that ABFT overhead is usually smaller
than Chen's verification cost.

We expect this combined approach to be interesting for other variants of the
preconditioned conjugate gradient algorithm~\cite{saad}.
Triangular preconditioners seem to be particularly attracting, in that it looks
possible to treat them by adapting the techniques described in this paper
(Shantharam et al.~\cite{Shantharam2012} addressed the triangular case).

\section*{Acknowledgements}
Y.~Robert and B.~U\c{c}ar were partly supported by the French Research Agency
(ANR) through the Rescue and SOLHAR (ANR MONU-13-0007) projects.  Y.~Robert is
with Institut Universitaire de France. J.~Langou was fully supported by NSF
award CCF 1054864.

\section*{References}

\bibliography{abftSpMV}

\end{document}

%% file: figures/comparison_341.tikz
% This file was created by matlab2tikz v0.4.7 running on MATLAB 8.5.
% Copyright (c) 2008--2014, Nico Schlömer <nico.schloemer@gmail.com>
% All rights reserved.
% Minimal pgfplots version: 1.3
%
% The latest updates can be retrieved from
%   http://www.mathworks.com/matlabcentral/fileexchange/22022-matlab2tikz
% where you can also make suggestions and rate matlab2tikz.
%
\begin{tikzpicture}

\begin{axis}[%
width=3cm,
height=3cm,
compat=newest,
plot coordinates/math parser=false,
trim axis left,
scale only axis,
separate axis lines,
every outer x axis line/.append style={white!15!black},
every x tick label/.append style={font=\color{white!15!black}},
xmode=log,
xmin=5,
xmax=50000,
xminorticks=true,
xmajorgrids,
xminorgrids,
every outer y axis line/.append style={white!15!black},
every y tick label/.append style={font=\color{white!15!black}},
yticklabel=\parbox{\widthof{200}}{\flushright\pgfmathprintnumber{\tick}},
ymin=0,
ymax=250.7069,
ymajorgrids,
legend style={draw=white!15!black,fill=white,legend cell align=left}
]
\addplot [color=green,solid,line width=2.0pt,forget plot]
  table[row sep=crcr]{5	110.8497\\
13.9127970110356	193.584\\
38.7131841340563	250.7069\\
107.721734501594	240.3342\\
299.74212515947	232.3217\\
834.05026860003	230.4736\\
2320.79441680639	229.1457\\
6457.74832507442	230.497\\
17969.0683190231	231.589\\
50000	229.2124\\
};
\addplot [color=blue,dashed,line width=2.0pt,forget plot]
  table[row sep=crcr]{5	66.24496\\
13.9127970110356	55.69611\\
38.7131841340563	53.05199\\
107.721734501594	53.35719\\
299.74212515947	53.1016\\
834.05026860003	53.08722\\
2320.79441680639	53.25253\\
6457.74832507442	53.84089\\
17969.0683190231	54.17459\\
50000	54.01631\\
};
\addplot [color=red,dotted,line width=2.0pt,forget plot]
  table[row sep=crcr]{5	85.57285\\
13.9127970110356	91.70231\\
38.7131841340563	92.52328\\
107.721734501594	91.4892\\
299.74212515947	90.73405\\
834.05026860003	93.73865\\
2320.79441680639	94.25548\\
6457.74832507442	94.66757\\
17969.0683190231	94.43095\\
50000	93.79567\\
};
\end{axis}
\end{tikzpicture}%

%% file: figures/comparison_752.tikz
% This file was created by matlab2tikz v0.4.7 running on MATLAB 8.5.
% Copyright (c) 2008--2014, Nico Schlömer <nico.schloemer@gmail.com>
% All rights reserved.
% Minimal pgfplots version: 1.3
% 
% The latest updates can be retrieved from
%   http://www.mathworks.com/matlabcentral/fileexchange/22022-matlab2tikz
% where you can also make suggestions and rate matlab2tikz.
% 
\begin{tikzpicture}

\begin{axis}[%
width=3cm,
height=3cm,
scale only axis,
separate axis lines,
every outer x axis line/.append style={white!15!black},
every x tick label/.append style={font=\color{white!15!black}},
yticklabel=\parbox{\widthof{200}}{\flushright\pgfmathprintnumber{\tick}},
xmode=log,
xmin=5,
xmax=50000,
xminorticks=true,
xmajorgrids,
xminorgrids,
every outer y axis line/.append style={white!15!black},
every y tick label/.append style={font=\color{white!15!black}},
ymin=0,
ymax=5.163034,
ymajorgrids,
legend style={draw=white!15!black,fill=white,legend cell align=left}
]
\addplot [color=green,solid,line width=2.0pt,forget plot]
  table[row sep=crcr]{5	4.470645\\
13.9127970110356	2.772222\\
38.7131841340563	2.582038\\
107.721734501594	2.428736\\
299.74212515947	2.397215\\
834.05026860003	2.424055\\
2320.79441680639	2.378444\\
6457.74832507442	2.427203\\
17969.0683190231	2.3792\\
50000	2.405181\\
};
\addplot [color=blue,dashed,line width=2.0pt,forget plot]
  table[row sep=crcr]{5	3.727595\\
13.9127970110356	2.460942\\
38.7131841340563	2.20191\\
107.721734501594	2.228859\\
299.74212515947	2.201856\\
834.05026860003	2.197856\\
2320.79441680639	2.198413\\
6457.74832507442	2.201384\\
17969.0683190231	2.201223\\
50000	2.200705\\
};
\addplot [color=red,dotted,line width=2.0pt,forget plot]
  table[row sep=crcr]{5	5.163034\\
13.9127970110356	3.946561\\
38.7131841340563	2.494251\\
107.721734501594	2.17914\\
299.74212515947	2.09878\\
834.05026860003	2.102445\\
2320.79441680639	2.100392\\
6457.74832507442	2.100209\\
17969.0683190231	2.097353\\
50000	2.100097\\
};
\end{axis}
\end{tikzpicture}%

%% file: figures/comparison_924.tikz
% This file was created by matlab2tikz v0.4.7 running on MATLAB 8.5.
% Copyright (c) 2008--2014, Nico Schlömer <nico.schloemer@gmail.com>
% All rights reserved.
% Minimal pgfplots version: 1.3
% 
% The latest updates can be retrieved from
%   http://www.mathworks.com/matlabcentral/fileexchange/22022-matlab2tikz
% where you can also make suggestions and rate matlab2tikz.
% 
\begin{tikzpicture}

\begin{axis}[%
width=3cm,
height=3cm,
scale only axis,
separate axis lines,
every outer x axis line/.append style={white!15!black},
every x tick label/.append style={font=\color{white!15!black}},
xmode=log,
xmin=5,
xmax=50000,
xminorticks=true,
xmajorgrids,
xminorgrids,
every outer y axis line/.append style={white!15!black},
every y tick label/.append style={font=\color{white!15!black}},
yticklabel=\parbox{\widthof{200}}{\flushright\pgfmathprintnumber{\tick}},
ymin=0,
ymax=49.86397,
ymajorgrids,
legend style={draw=white!15!black,fill=white,legend cell align=left}
]
\addplot [color=green,solid,line width=2.0pt,forget plot]
  table[row sep=crcr]{5	45.76235\\
13.9127970110356	30.71469\\
38.7131841340563	20.47237\\
107.721734501594	20.2565\\
299.74212515947	18.90951\\
834.05026860003	20.10616\\
2320.79441680639	18.80212\\
6457.74832507442	18.81354\\
17969.0683190231	18.80439\\
50000	18.81719\\
};
\addplot [color=blue,dashed,line width=2.0pt,forget plot]
  table[row sep=crcr]{5	19.91506\\
13.9127970110356	18.09958\\
38.7131841340563	18.0521\\
107.721734501594	13.97381\\
299.74212515947	13.88631\\
834.05026860003	14.51295\\
2320.79441680639	14.06095\\
6457.74832507442	13.57293\\
17969.0683190231	13.57647\\
50000	13.56314\\
};
\addplot [color=red,dotted,line width=2.0pt,forget plot]
  table[row sep=crcr]{5	49.86397\\
13.9127970110356	38.37939\\
38.7131841340563	13.69599\\
107.721734501594	13.76709\\
299.74212515947	13.27663\\
834.05026860003	12.93338\\
2320.79441680639	12.94337\\
6457.74832507442	12.94328\\
17969.0683190231	12.93829\\
50000	12.94006\\
};
\end{axis}
\end{tikzpicture}%

%% file: figures/comparison_1288.tikz
% This file was created by matlab2tikz v0.4.7 running on MATLAB 8.5.
% Copyright (c) 2008--2014, Nico Schlömer <nico.schloemer@gmail.com>
% All rights reserved.
% Minimal pgfplots version: 1.3
%
% The latest updates can be retrieved from
%   http://www.mathworks.com/matlabcentral/fileexchange/22022-matlab2tikz
% where you can also make suggestions and rate matlab2tikz.
%
\begin{tikzpicture}

\begin{axis}[%
width=3cm,
height=3cm,
scale only axis,
separate axis lines,
every outer x axis line/.append style={white!15!black},
every x tick label/.append style={font=\color{white!15!black}},
xmode=log,
xmin=5,
xmax=50000,
xminorticks=true,
xmajorgrids,
xminorgrids,
every outer y axis line/.append style={white!15!black},
every y tick label/.append style={font=\color{white!15!black}},
yticklabel=\parbox{\widthof{200}}{\flushright\pgfmathprintnumber{\tick}},
ymin=0,
ymax=6.825187,
ymajorgrids,
legend style={draw=white!15!black,fill=white,legend cell align=left}
]
\addplot [color=green,solid,line width=2.0pt,forget plot]
  table[row sep=crcr]{5	4.760459\\
13.9127970110356	3.806762\\
38.7131841340563	3.007971\\
107.721734501594	3.203419\\
299.74212515947	2.862017\\
834.05026860003	2.862112\\
2320.79441680639	2.85504\\
6457.74832507442	2.858985\\
17969.0683190231	2.856108\\
50000	2.860831\\
};
\addplot [color=blue,dashed,line width=2.0pt,forget plot]
  table[row sep=crcr]{5	4.996113\\
13.9127970110356	3.67325\\
38.7131841340563	2.79205\\
107.721734501594	2.760473\\
299.74212515947	2.760369\\
834.05026860003	2.761046\\
2320.79441680639	2.756594\\
6457.74832507442	2.761188\\
17969.0683190231	2.759242\\
50000	2.758721\\
};
\addplot [color=red,dotted,line width=2.0pt,forget plot]
  table[row sep=crcr]{5	6.825187\\
13.9127970110356	3.855426\\
38.7131841340563	3.42224\\
107.721734501594	2.664788\\
299.74212515947	2.740511\\
834.05026860003	2.598266\\
2320.79441680639	2.594399\\
6457.74832507442	2.599092\\
17969.0683190231	2.59557\\
50000	2.59562\\
};
\end{axis}
\end{tikzpicture}%

%% file: figures/comparison_1289.tikz
% This file was created by matlab2tikz v0.4.7 running on MATLAB 8.5.
% Copyright (c) 2008--2014, Nico Schlömer <nico.schloemer@gmail.com>
% All rights reserved.
% Minimal pgfplots version: 1.3
% 
% The latest updates can be retrieved from
%   http://www.mathworks.com/matlabcentral/fileexchange/22022-matlab2tikz
% where you can also make suggestions and rate matlab2tikz.
% 
\begin{tikzpicture}

\begin{axis}[%
width=3cm,
height=3cm,
scale only axis,
separate axis lines,
every outer x axis line/.append style={white!15!black},
every x tick label/.append style={font=\color{white!15!black}},
xmode=log,
xmin=5,
xmax=50000,
xminorticks=true,
xmajorgrids,
xminorgrids,
every outer y axis line/.append style={white!15!black},
every y tick label/.append style={font=\color{white!15!black}},
yticklabel=\parbox{\widthof{200}}{\flushright\pgfmathprintnumber{\tick}},
ymin=0,
ymax=10.253385,
ymajorgrids,
legend style={draw=white!15!black,fill=white,legend cell align=left}
]
\addplot [color=green,solid,line width=2.0pt,forget plot]
  table[row sep=crcr]{5	8.325613\\
13.9127970110356	7.618277\\
38.7131841340563	7.38914\\
107.721734501594	7.478978\\
299.74212515947	7.508952\\
834.05026860003	7.070876\\
2320.79441680639	7.089236\\
6457.74832507442	7.103107\\
17969.0683190231	6.992414\\
50000	6.983612\\
};
\addplot [color=blue,dashed,line width=2.0pt,forget plot]
  table[row sep=crcr]{5	7.529211\\
13.9127970110356	5.141338\\
38.7131841340563	5.211812\\
107.721734501594	4.599863\\
299.74212515947	4.960965\\
834.05026860003	4.926715\\
2320.79441680639	4.880859\\
6457.74832507442	4.875986\\
17969.0683190231	4.946716\\
50000	4.85486\\
};
\addplot [color=red,dotted,line width=2.0pt,forget plot]
  table[row sep=crcr]{5	10.253385\\
13.9127970110356	5.852135\\
38.7131841340563	3.176506\\
107.721734501594	3.143744\\
299.74212515947	2.905742\\
834.05026860003	2.927444\\
2320.79441680639	2.880793\\
6457.74832507442	2.889643\\
17969.0683190231	2.875488\\
50000	2.871864\\
};
\end{axis}
\end{tikzpicture}%

%% file: figures/comparison_1311.tikz
% This file was created by matlab2tikz v0.4.7 running on MATLAB 8.5.
% Copyright (c) 2008--2014, Nico Schlömer <nico.schloemer@gmail.com>
% All rights reserved.
% Minimal pgfplots version: 1.3
% 
% The latest updates can be retrieved from
%   http://www.mathworks.com/matlabcentral/fileexchange/22022-matlab2tikz
% where you can also make suggestions and rate matlab2tikz.
% 
\begin{tikzpicture}

\begin{axis}[%
width=3cm,
height=3cm,
scale only axis,
separate axis lines,
every outer x axis line/.append style={white!15!black},
every x tick label/.append style={font=\color{white!15!black}},
xmode=log,
xmin=5,
xmax=50000,
xminorticks=true,
xmajorgrids,
xminorgrids,
every outer y axis line/.append style={white!15!black},
every y tick label/.append style={font=\color{white!15!black}},
yticklabel=\parbox{\widthof{200}}{\flushright\pgfmathprintnumber{\tick}},
ymin=0,
ymax=229.9679,
ymajorgrids,
legend style={draw=white!15!black,fill=white,legend cell align=left}
]
\addplot [color=green,solid,line width=2.0pt,forget plot]
  table[row sep=crcr]{5	106.5526\\
13.9127970110356	206.9635\\
38.7131841340563	229.9679\\
107.721734501594	184.835\\
299.74212515947	180.5717\\
834.05026860003	159.9832\\
2320.79441680639	160.3948\\
6457.74832507442	160.9335\\
17969.0683190231	160.4228\\
50000	161.1698\\
};
\addplot [color=blue,dashed,line width=2.0pt,forget plot]
  table[row sep=crcr]{5	39.46612\\
13.9127970110356	36.45868\\
38.7131841340563	33.99827\\
107.721734501594	35.43228\\
299.74212515947	35.83721\\
834.05026860003	36.1478\\
2320.79441680639	35.94009\\
6457.74832507442	36.16732\\
17969.0683190231	36.11694\\
50000	36.42013\\
};
\addplot [color=red,dotted,line width=2.0pt,forget plot]
  table[row sep=crcr]{5	59.27672\\
13.9127970110356	64.74738\\
38.7131841340563	65.93721\\
107.721734501594	66.04662\\
299.74212515947	67.37543\\
834.05026860003	67.38313\\
2320.79441680639	67.33386\\
6457.74832507442	67.31405\\
17969.0683190231	67.37975\\
50000	67.61867\\
};
\end{axis}
\end{tikzpicture}%

%% file: figures/comparison_1312.tikz
% This file was created by matlab2tikz v0.4.7 running on MATLAB 8.5.
% Copyright (c) 2008--2014, Nico Schlömer <nico.schloemer@gmail.com>
% All rights reserved.
% Minimal pgfplots version: 1.3
% 
% The latest updates can be retrieved from
%   http://www.mathworks.com/matlabcentral/fileexchange/22022-matlab2tikz
% where you can also make suggestions and rate matlab2tikz.
% 
\begin{tikzpicture}

\begin{axis}[%
width=3cm,
height=3cm,
scale only axis,
separate axis lines,
every outer x axis line/.append style={white!15!black},
every x tick label/.append style={font=\color{white!15!black}},
xmode=log,
xmin=5,
xmax=50000,
xminorticks=true,
xmajorgrids,
xminorgrids,
every outer y axis line/.append style={white!15!black},
every y tick label/.append style={font=\color{white!15!black}},
yticklabel=\parbox{\widthof{200}}{\flushright\pgfmathprintnumber{\tick}},
ymin=0,
ymax=8.205841,
ymajorgrids,
legend style={draw=white!15!black,fill=white,legend cell align=left}
]
\addplot [color=green,solid,line width=2.0pt,forget plot]
  table[row sep=crcr]{5	8.205841\\
13.9127970110356	4.819749\\
38.7131841340563	4.626551\\
107.721734501594	3.622984\\
299.74212515947	3.595087\\
834.05026860003	3.591684\\
2320.79441680639	3.597549\\
6457.74832507442	3.592552\\
17969.0683190231	3.594649\\
50000	3.596816\\
};
\addplot [color=blue,dashed,line width=2.0pt,forget plot]
  table[row sep=crcr]{5	5.000844\\
13.9127970110356	3.590721\\
38.7131841340563	3.054736\\
107.721734501594	2.792486\\
299.74212515947	2.740951\\
834.05026860003	2.739635\\
2320.79441680639	2.740963\\
6457.74832507442	2.739221\\
17969.0683190231	2.73822\\
50000	2.739692\\
};
\addplot [color=red,dotted,line width=2.0pt,forget plot]
  table[row sep=crcr]{5	5.835808\\
13.9127970110356	5.523221\\
38.7131841340563	4.166796\\
107.721734501594	2.622847\\
299.74212515947	2.617927\\
834.05026860003	2.619285\\
2320.79441680639	2.618421\\
6457.74832507442	2.676352\\
17969.0683190231	2.617823\\
50000	2.617789\\
};
\end{axis}
\end{tikzpicture}%

%% file: figures/comparison_1848.tikz
% This file was created by matlab2tikz v0.4.7 running on MATLAB 8.5.
% Copyright (c) 2008--2014, Nico Schlömer <nico.schloemer@gmail.com>
% All rights reserved.
% Minimal pgfplots version: 1.3
% 
% The latest updates can be retrieved from
%   http://www.mathworks.com/matlabcentral/fileexchange/22022-matlab2tikz
% where you can also make suggestions and rate matlab2tikz.
% 
\begin{tikzpicture}

\begin{axis}[%
width=3cm,
height=3cm,
scale only axis,
separate axis lines,
every outer x axis line/.append style={white!15!black},
every x tick label/.append style={font=\color{white!15!black}},
xmode=log,
xmin=5,
xmax=50000,
xminorticks=true,
xmajorgrids,
xminorgrids,
every outer y axis line/.append style={white!15!black},
every y tick label/.append style={font=\color{white!15!black}},
yticklabel=\parbox{\widthof{200}}{\flushright\pgfmathprintnumber{\tick}},
ymin=0,
ymax=143.7932,
ymajorgrids,
legend style={draw=white!15!black,fill=white,legend cell align=left}
]
\addplot [color=green,solid,line width=2.0pt,forget plot]
  table[row sep=crcr]{5	143.7932\\
13.9127970110356	107.1593\\
38.7131841340563	69.98727\\
107.721734501594	66.53503\\
299.74212515947	64.65892\\
834.05026860003	64.50634\\
2320.79441680639	65.41817\\
6457.74832507442	63.95319\\
17969.0683190231	64.3641\\
50000	63.98372\\
};
\addplot [color=blue,dashed,line width=2.0pt,forget plot]
  table[row sep=crcr]{5	41.90268\\
13.9127970110356	33.26128\\
38.7131841340563	26.43955\\
107.721734501594	22.53363\\
299.74212515947	18.6822\\
834.05026860003	18.64246\\
2320.79441680639	18.50986\\
6457.74832507442	18.11247\\
17969.0683190231	18.07819\\
50000	18.1938\\
};
\addplot [color=red,dotted,line width=2.0pt,forget plot]
  table[row sep=crcr]{5	74.54383\\
13.9127970110356	45.00674\\
38.7131841340563	21.78333\\
107.721734501594	22.0833\\
299.74212515947	19.96332\\
834.05026860003	19.97182\\
2320.79441680639	19.92554\\
6457.74832507442	19.86522\\
17969.0683190231	20.0004\\
50000	19.8685\\
};
\end{axis}
\end{tikzpicture}%

%% file: figures/comparison_2213.tikz
% This file was created by matlab2tikz v0.4.7 running on MATLAB 8.5.
% Copyright (c) 2008--2014, Nico Schlömer <nico.schloemer@gmail.com>
% All rights reserved.
% Minimal pgfplots version: 1.3
% 
% The latest updates can be retrieved from
%   http://www.mathworks.com/matlabcentral/fileexchange/22022-matlab2tikz
% where you can also make suggestions and rate matlab2tikz.
% 
\begin{tikzpicture}

\begin{axis}[%
width=3cm,
height=3cm,
scale only axis,
separate axis lines,
every outer x axis line/.append style={white!15!black},
every x tick label/.append style={font=\color{white!15!black}},
yticklabel=\parbox{\widthof{200}}{\flushright\pgfmathprintnumber{\tick}},
xmode=log,
xmin=5,
xmax=50000,
xminorticks=true,
xmajorgrids,
xminorgrids,
every outer y axis line/.append style={white!15!black},
every y tick label/.append style={font=\color{white!15!black}},
ymin=0,
ymax=1.196284,
ymajorgrids,
legend style={draw=white!15!black,fill=white,legend cell align=left}
]
\addplot [color=green,solid,line width=2.0pt,forget plot]
  table[row sep=crcr]{5	0.5694746\\
13.9127970110356	0.5126294\\
38.7131841340563	0.494519\\
107.721734501594	0.4988982\\
299.74212515947	0.4974604\\
834.05026860003	0.4939344\\
2320.79441680639	0.4930246\\
6457.74832507442	0.4959968\\
17969.0683190231	0.494784\\
50000	0.4971064\\
};
\addplot [color=blue,dashed,line width=2.0pt,forget plot]
  table[row sep=crcr]{5	0.7601928\\
13.9127970110356	0.581854\\
38.7131841340563	0.420455\\
107.721734501594	0.4587214\\
299.74212515947	0.457309\\
834.05026860003	0.4574682\\
2320.79441680639	0.4578796\\
6457.74832507442	0.4571722\\
17969.0683190231	0.4578532\\
50000	0.4583394\\
};
\addplot [color=red,dotted,line width=2.0pt,forget plot]
  table[row sep=crcr]{5	1.196284\\
13.9127970110356	0.4655646\\
38.7131841340563	0.465943\\
107.721734501594	0.469014\\
299.74212515947	0.4667292\\
834.05026860003	0.4682564\\
2320.79441680639	0.4655632\\
6457.74832507442	0.4653268\\
17969.0683190231	0.4669068\\
50000	0.4659178\\
};
\end{axis}
\end{tikzpicture}%